\documentclass[10pt,journal,compsoc]{IEEEtran}
\usepackage{cite}
\usepackage{graphicx}
\usepackage[cmex10]{amsmath}
\usepackage{amssymb}
\usepackage{amsthm}
\usepackage{array}
\usepackage{url}
\usepackage{algorithm,algorithmic}
\usepackage{chemarrow}
\usepackage{multirow}
\usepackage{extarrows}
\usepackage{setspace}
\usepackage{booktabs}
\usepackage{tabulary}
\usepackage[shortlabels]{enumitem}
\usepackage{balance}
\usepackage[cmex10]{amsmath}
\usepackage[shortlabels]{enumitem}
\usepackage{amssymb}
\usepackage{amsthm}
\usepackage{multirow}
\usepackage{cite}
\usepackage{color}
\usepackage{setspace}
\usepackage{ragged2e}
\usepackage{stmaryrd} 
\usepackage{tabularx}
\usepackage{makecell}
\usepackage{footnote}
\usepackage{threeparttable}
\usepackage{eqparbox}
\usepackage{bm}
\usepackage{bbding}

\makeatletter
\newcommand{\ALOOP}[1]{\ALC@it\algorithmicloop\ #1
	\begin{ALC@loop}}
	\newcommand{\ENDALOOP}{\end{ALC@loop}\ALC@it\algorithmicendloop}

\newtheorem{theorem}{\textbf{\emph{Theorem}}}
\newtheorem{definition}{\textbf{\emph{Definition}}}
\newcommand{\main}{PrivGED}
\newcommand{\eig}{eigendecomposition}
\newcommand{\eigs}{eigenvalues/eigenvectors}
\newcommand{\cs}{$\mathcal{CS}_{1}$}
\newcommand{\css}{$\mathcal{CS}_{2}$}
\newcommand{\csa}{$\mathcal{CS}_{\{1,2\}}$}

\begin{document}
	
\title{Privacy-Preserving Analytics on Decentralized Social Graphs: The Case of Eigendecomposition}
	
\author{Songlei Wang, Yifeng Zheng, Xiaohua Jia, \IEEEmembership{Fellow, IEEE}, and Xun Yi 
		
		\IEEEcompsocitemizethanks{
			\IEEEcompsocthanksitem Songlei Wang and Yifeng Zheng are with the School of Computer Science and Technology, Harbin Institute of Technology, Shenzhen, Guangdong 518055, China (e-mail: songlei.wang@outlook.com, yifeng.zheng@hit.edu.cn).
			\IEEEcompsocthanksitem Xiaohua Jia is with the School of Computer Science and Technology, Harbin Institute of Technology, Shenzhen, China, and also with the Department of Computer Science, City University of Hong Kong, Kowloon Tong, Hong Kong, China (e-mail: csjia@cityu.edu.hk).
			\IEEEcompsocthanksitem Xun Yi is with the School of Computing Technologies, RMIT University, Melbourne, Australia (e-mail: xun.yi@rmit.edu.au).
			\IEEEcompsocthanksitem Corresponding author: Yifeng Zheng.
		}
	}
	\IEEEtitleabstractindextext{
		\begin{abstract}
		Analytics over social graphs allows to extract valuable knowledge and insights for many fields like community detection, fraud detection, and interest mining. In practice, decentralized social graphs frequently arise, where the social graph is not available to a single entity and is decentralized among a large number of users, each holding only a limited local view about the whole graph. Collecting the local views for analytics of decentralized social graphs raises critical privacy concerns, as they encode private information about the social interactions among individuals. In this paper, we design, implement, and evaluate PrivGED, a new system aimed at privacy-preserving analytics over decentralized social graphs. PrivGED focuses on the support for eigendecomposition, one popular and fundamental graph analytics task producing eigenvalues/eigenvectors over the adjacency matrix of a social graph and benefits various practical applications. PrivGED is built from a delicate synergy of insights on graph analytics, lightweight cryptography, and differential privacy, allowing users to securely contribute their local views on a decentralized social graph for a cloud-based eigendecomposition analytics service while gaining strong privacy protection. Extensive experiments over real-world social graph datasets demonstrate that PrivGED achieves accuracy comparable to the plaintext domain, with practically affordable performance superior to prior art.
		\end{abstract}
		
	}

	\maketitle

	\IEEEdisplaynontitleabstractindextext
	
	\IEEEpeerreviewmaketitle
	
	\section{Introduction}
	\label{sec:Intro}
	\textbf{}

	%
	\IEEEPARstart{A}nalytics over information-rich social graphs allows the extraction of valuable and impactful knowledge and insights for many fields like community detection, fraud detection, and interest mining \cite{Tabassum2018Social,UmitPhysica2019}.
	%
	Social graph analytics, however, becomes quite challenging when the social graph is not available to a single entity and presented in a \emph{decentralized} manner, where each user only holds a limited local view about the whole social graph, and the complete social graph is formed by their collective views.
	%
	Decentralized social graphs can arise in many practical applications \cite{SunXKYQWY19,qin2017generating,xue2011p3d,zhang2014message,ma2018armor}.
	For example, in a phone network, each user has his own contact list and the collective contact lists of all users form a social graph in a decentralized manner \cite{qin2017generating}.

	Collecting individual users' local views for analytics in the setting of decentralized social graphs can raise critical privacy concerns, as these local views encode sensitive information regarding the social interactions among individuals \cite{qin2017generating,sharma2018privategraph}.
	Users thus may be reluctant to be engaged in such analytics if their local views gain no protection.
	Thus, it is of critical importance to ensure that security must be embedded in analytics over decentralized social graphs from the very beginning, so that valuable knowledge and insights can be extracted without compromising the privacy of individual users.
	Among others, one popular and fundamental graph analytics task is {\eig} which we focus on as a concrete instantiation in this paper.
	Eigendecomposition-based social graph analytics works on the adjacency matrix associated with a social graph to yield {\eigs}, and can benefit various applications, such as community structure detection \cite{newman2013spectral}, important members finding \cite{wang2011identifying}, and social graph partitioning \cite{newman2013spectral} (see Section \ref{subsec:applications of graph analytics via eig} for more details on applications).

	In the literature, little work \cite{wang2013differential,ahmed2020publishing,sharma2018privategraph} has been done regarding privacy-preserving {\eig} on graphs. 
	Some works \cite{wang2013differential,ahmed2020publishing} focus on publishing adjacency matrices with differential privacy while preserving their {\eigs}. 
	Yet these works operate with \emph{centralized} social graphs, where the social graph is available to a single entity and processed in the plaintext domain.
	%
	%
	The most related (state-of-the-art) work to ours is due to Sharma \textit{et al.} \cite{sharma2018privategraph}, who propose a method PrivateGraph that works under a decentralized social graph setting and aims to provide privacy protection for individuals' local views.
	%
	However, PrivateGraph is not quite satisfactory due to the following downsides in functionality and security.
	Firstly, PrivateGraph only supports {\eig} on undirected graphs (via the Lanczos method \cite{lanczos1950iteration}), but many social graphs in practice are directed \cite{curtiss2013unicorn}, which cannot be supported via the Lanczos method.
	Secondly, it requires some users to expose the number of their friends, posing a threat to their privacy \cite{zhou2008brief}.
	Thirdly, it requires frequent online interactions between the cloud that coordinates the \eig~based analytics task and the entity that requests the {\eigs}.
	%
	%
	Therefore, how to achieve privacy-preserving {\eig}-based analytics over decentralized social graphs is still challenging and remains to be fully explored.

	In light of the above, in this paper, we design, implement, and evaluate {\main}, a new system that allows privacy-preserving analytics over decentralized social graphs with \eig.
	Leveraging the emerging cloud-empowered graph analytics paradigm, {\main} allows a set of users to securely contribute their local views on a decentralized social graph for an \eig~analytics service empowered by the cloud, while ensuring strong protection on individual local views.

	We start with considering how to enable the individual local views to be securely collected so as to form the (encrypted) adjacency matrix adequately for \eig~on the cloud.
	Each row vector in the adjacency matrix stores information of the local view of a user.
	Targeting security assurance as well as high efficiency, \main~resorts to a lightweight cryptographic technique---additive secret sharing (ASS) \cite{mohassel2017secureml}, for fast encryption of the elements in local view vectors.
	However, simply applying ASS over each user's complete vector is inefficient because social graphs are usually large-scale and sparse \cite{curtiss2013unicorn}, leading to many zero elements in local view vectors that incur undesirable performance overheads.
	To tackle this problem, {\main} develops techniques that allow to exploit the benefits of graph sparsity for efficiency while protecting the privacy of users' private social relationships, through a delicate synergy of sparse representation, local differential privacy (LDP) \cite{qin2017generating}, and function secret sharing (FSS) \cite{boyle2016function,boyle2021function} techniques. 
	%
	%
	%
	As opposed to PrivateGraph \cite{sharma2018privategraph}, \main~does not reveal any users' exact node degree information.

	Subsequently, we consider how to enable {\eig} to be securely performed on the cloud over the formed encrypted adjacency matrix.
	We first make two important practical observations: i) Usually only the top-$k$ {\eigs} ($k\ll N$; $N$ is the number of users and determines the size of the adjacency matrix) are needed in practice \cite{newman2006finding, wang2011identifying, newman2013spectral,kamvar2003extrapolation,kamvar2004adaptive,charalambous2016totally, sangers2015eigenvectors}; ii) The desired top-$k$ \eigs~can be derived from the complete \eigs~of a smaller matrix reduced from the original matrix under adequate dimension reduction methods, among which the most popular ones are the Arnoldi method \cite{arnoldi1951principle} (for general (possibly non-symmetric) matrices) and the Lanczos method \cite{lanczos1950iteration} (for symmetric matrices).
	Therefore, \main~first introduces effective techniques, which tackle the challenging squared root and division operations in the secret sharing domain, to securely realize the Arnoldi method and the Lanczos method.
	This allows {\main} to flexibly work on both undirected and directed graphs.
	{\main} then further provides techniques for secure realization of the widely used QR algorithm \cite{francis1962qr} over the dimension-reduced matrix so as to produce the encrypted desired \eigs.
	We reformulate the plaintext QR algorithm to ease computation in the ciphertext domain as well as optimize the processing of secret-shared matrix multiplications for high performance.
	We highlight our main contributions below:
	
	\begin{itemize}
		\item We present {\main}, a new system supporting privacy-preserving analytics over decentralized social graphs with eigendecomposition.
		
		\item We develop techniques for secure collection of individual local views on the decentralized social graph, which exploit the benefits of graph sparsity for efficiency while protecting the privacy of individual's social relationships.

		
		\item We develop techniques for securely realizing the Arnoldi/Lanczos methods and the QR algorithm, so as to fully support the processing pipeline of secure \eig~on the cloud.

		\item We formally analyze the security of {\main}, implement it with $\sim$2000 lines of Python code, and conduct an extensive evaluation over three real-world datasets. The results demonstrate that {\main} achieves accuracy comparable to the plaintext domain, with practically affordable performance superior to prior art.
	\end{itemize}

	The rest of this paper is organized as follows. Section \ref{sec:related_work} discusses the related work. Section \ref{sec:pre} introduces preliminaries. Section \ref{sec:problem_def} presents the problem statement. Section \ref{sec:main} and Section \ref{sec:secure_eig} give the detailed design. The privacy and security analysis is presented in Section \ref{sec:security_analy}. We present experiment results in Section \ref{sec:experiments} and conclude this paper in Section \ref{sec:conclusion}.
	
	\section{Related Work}
	\label{sec:related_work}
	
	\subsection{Graph Analytics via Eigendecomposition}
	\label{subsec:applications of graph analytics via eig}
	Graphs can characterize the complex inter-dependency among entities, and are used in various applications, such as social networks \cite{curtiss2013unicorn} and webpage networks \cite{kamvar2003extrapolation}. As one popular and fundamental graph analytics task, {\eig} works on the adjacency matrix to yield {\eigs}, and can benefit various applications \cite{newman2006finding, wang2011identifying, newman2013spectral,page1999pagerank,kamvar2003extrapolation, kamvar2004adaptive,charalambous2016totally}. 
	For example, \eig-based graph analytics can greatly benefit community detection through the following ways: 1) finding the community structure based on the eigenvectors \cite{newman2006finding}; 2) identifying and characterizing nodes importance to the community according to the relative change in the eigenvalues after removing them \cite{wang2011identifying}; 3) partitioning a social graph based on its top-2 eigenvectors \cite{newman2013spectral}.
	Another important application is PageRank \cite{page1999pagerank,kamvar2003extrapolation, kamvar2004adaptive,charalambous2016totally}, which is one of the best-known ranking algorithms in web search. PageRank measures the importance of website pages by computing the principal eigenvector of the matrix describing the hyperlinks of the website pages. In addition, the second eigenvector can be used to detect a certain type of link spamming \cite{sangers2015eigenvectors}. However, all of them consider the execution of {\eig} in the plaintext domain without privacy protection.

	\subsection{Privacy-Preserving Graph Analytics}
	\label{sec:relet1}
	
	There exist a variety of designs that aim to securely perform certain graph analytics tasks. 
	Some works \cite{MengRL21,zhou2020vertically,wu2021fedgnn} focus on privacy-preserving training of graph neural networks (GNN) based on the federated learning paradigm \cite{zhang2019deeppar}, which aim to train GNN models across multiple clients holding local datasets (e.g., spatio-temporal datasets \cite{MengRL21} or graphs \cite{zhou2020vertically,wu2021fedgnn}) in such a way that the datasets are kept local.
	The work \cite{MengRL21} focuses on GNNs over decentralized spatio-temporal data, and has the clients exchange model updates with the server in cleartext.
	%
	In contrast, the works \cite{zhou2020vertically,wu2021fedgnn} focus on graph datasets and design privacy-preserving mechanisms to protect the individual model updates.

	There has been a line of work \cite{wu2016privacy,du2020graphshield,araki2021secure,ShenMZMDH18,LiuZHC21} aimed at the support for graph analytics with cryptographic methods like secure multi-party computation techniques and searchable encryption.
	The main focus of this line of work has been on supporting different kinds of graph queries under different scenarios in a secure manner.
	Wu \textit{et al.} \cite{wu2016privacy} propose a protocol for privacy-preserving shorted path query in a two-party setting, where a client holding a query and a server holding a plaintext graph, based on cryptographic techniques including private information retrieval, garbled circuits, and oblivious transfer.
	In contrast, the works \cite{ShenMZMDH18,LiuZHC21} focus on the support for privacy-preserving shortest path queries in an outsourcing setting, where private queries need to be executed over encrypted graphs outsourced to the cloud.
	The work \cite{du2020graphshield} designs protocols that can support privacy-preserving shortest distance queries and maximum flow queries over outsourced encrypted graphs.
	These works rely on the combination of searchable encryption (e.g., order-revealing encryption), homomorphic encryption, and/or garbled circuits.
	In \cite{araki2021secure}, Araki \textit{et al.} consider a scenario where all nodes and edges of a graph are secret-shared between three servers and devise protocols for breadth-first search and maximal independent set queries, based on secret sharing and secure shuffling.
	The above works all target graph analytics tasks that are substantially different from the one considered in this paper.

	\begin{table}[t!]
		\small
		\centering
		\caption{Comparison with the State-of-the-Art Work PrivateGraph \cite{sharma2018privategraph}}
		\label{Tab:DiffWithPrivateG}
		\begin{tabular*}{\hsize}{@{}@{\extracolsep{\fill}}c|c|c}
			\toprule
			Property & PrivateGraph \cite{sharma2018privategraph} & {\main}\\\hline
			Undirected graph supported &$\checkmark$& $\checkmark$\\
			Directed graph supported &$\times$&  $\checkmark$\\
			All users' degrees protected &$\times$& $\checkmark$\\
			Analyst allowed to stay offline&$\times$&  $\checkmark$\\
			Lightweight cryptography & $\times$& $\checkmark$\\
			\bottomrule
		\end{tabular*}
		\vspace{-10pt}
	\end{table}

	There is another line of work \cite{wang2013differential,ahmed2020publishing} focuses on publishing graph matrices with differential privacy while preserving their {\eigs}.
	They work under the setting of centralized social graphs where the social graph is held by a single entity and processed in the plaintext domain.
	Some works consider the scenario of decentralized social graphs, and focus on the privacy-preserving support for \emph{different} tasks with differential privacy, such as estimating subgraph counts \cite{SunXKYQWY19} and generating representative synthetic social graphs \cite{qin2017generating}.

    The state-of-the-art design that is most related to ours is PrivateGraph \cite{sharma2018privategraph}, which is also aimed at \eig~analytics over decentralized social graphs with privacy protection.
	However, as mentioned above, PrivateGraph is subject to several crucial downsides in terms of functionality and security, which greatly limit its practical usability.
	In light of this gap, we present a new system design \main~for privacy-preserving {\eig} analytics on decentralized social graphs.
	Compared to PrivateGraph, \main~is much advantageous in that it (i) supports both directed and undirected social graphs (via secure realizations of both the Arnoldi method and the Lanczos method), (ii) does not reveal any users' exact degree information, and (iii) fully exploits the cloud to free the analyst from staying online for active and frequent interactions and conducting a large amount of local intermediate and post processing. In particular, as reported in \cite{sharma2018privategraph}, to obtain the {\eigs}, the analyst must spend 0.2 hours as well as communicate 10 GB with the cloud. In contrast, {\main} allows the analyst to directly receive the final {\eigs}. Table \ref{Tab:DiffWithPrivateG} summarizes the prominent advantages of our {\main} over PrivateGraph.

	\section{Preliminaries}
	\label{sec:pre}
	
	\subsection{Eigendecomposition-based Graph Analytics}
	\label{sec:preSpectral}

	A graph comprises a set of nodes with a corresponding set of edges which connect the nodes. The edges may be directed or undirected and may have weights associated with them as well.
	Eigendecomposition works on the adjacency matrix associated with a graph to yield {\eigs}.
	%
	A complete {\eig} on an $N*N$ matrix poses a considerable time complexity of $O(N^{3})$, which indeed also results in unnecessary cost for a large $N$ since only top-$k$ ($k\ll N$) {\eigs} are used in most {\eig}-based graph analysis tasks \cite{newman2006finding, wang2011identifying, newman2013spectral,kamvar2003extrapolation,kamvar2004adaptive, sangers2015eigenvectors, charalambous2016totally}. 
	In practice, given a large-scale adjacency matrix $\mathbf{A}$, to calculate its top-$k$ {\eigs}, the first step is to reduce its dimension from $N*N$ to $M*M$ ($M$ is usually slightly larger than $k$), producing a new matrix $\overline{\mathbf{A}}$ for further processing. The most popular dimension reduction methods are the Arnoldi method (Algorithm \ref{algo:1}) \cite{arnoldi1951principle} and the Lanczos method (Algorithm \ref{algo:2}) \cite{lanczos1950iteration}, which work on general (possibly non-symmetric) matrices and symmetric matrices, respectively. After dimension reduction, the QR algorithm \cite{francis1962qr} is usually used to efficiently calculate the complete {\eigs} of $\overline{\mathbf{A}}$. Finally, the top-$k$ eigenvalues of $\overline{\mathbf{A}}$ are used to represent the top-$k$ eigenvalues of $\mathbf{A}$, and the corresponding eigenvectors $\overline{\mathbf{V}}$ can be transformed to the eigenvectors $\mathbf{V}$ of $\mathbf{A}$ by $\mathbf{V}=\mathbf{P}\overline{\mathbf{V}}$ where $\mathbf{P}$ is determined by line $\ref{alg:AQ}$, Algorithm \ref{algo:1} or line $\ref{alg:LQ}$, Algorithm \ref{algo:2}.

	\subsection{Local Differential Privacy}
	\label{sec:LDP}
	
	Compared to the traditional differential privacy model \cite{dwork2006differential} that assumes a trusted data collector which can collect and see raw data, the recently emerging LDP model \cite{kasiviswanathan2011can} considers the data collector to be untrusted, in which each user only reports perturbed data with calibrated noises added. The formal definition of $(\epsilon, \delta)$-LDP is as follows.
	\begin{definition}
		\label{def:LDP}
		A randomized mechanism $\mathcal{M}$ satisfies $(\epsilon, \delta)$-LDP, if and only if for any inputs $x$ and $x'$, we have: $\forall y\in Range(\mathcal{M})$,  
		\begin{equation}\notag
			\label{eq:LDP}
			Pr[\mathcal{M}(x)= y]\leq e^{\epsilon}\cdot Pr[\mathcal{M}(x')= y] +\delta,
		\end{equation}
		where $Range(\mathcal{M})$ denotes the set of all possible outputs of $\mathcal{M}$, $\epsilon$ is the privacy budget, and $\delta$ is a privacy parameter.
	\end{definition}

	Laplace distribution is a widely popular choice to draw the noises, which is formally defined as follows.
	
	\begin{definition}
		A discrete random variable $x$ follows $Lap(\epsilon,\delta,\Delta)$ distribution if its probability density function is \cite{he2017composing}
		\begin{equation}\notag
			Pr[x]=\frac{e^{\frac{\epsilon}{\Delta}}-1}{e^{\frac{\epsilon}{\Delta}}+1}\cdot e^{\frac{-\epsilon\cdot|x-\mu|}{\Delta}}, \forall x\in \mathbb{Z},
		\end{equation}
		where $\mu$ is the mean of the Laplace distribution: 
		\begin{equation}\label{eq:mu}
			\mu=-\frac{\Delta\cdot\ln[(e^{\frac{\epsilon}{\Delta}}+1)\cdot(1-(1-\delta)^{\frac{1}{\Delta}})]}{\epsilon}.
		\end{equation}
		$\Delta$ is the sensitivity of a function $f$:
		\begin{equation}\notag
			\Delta=max|f(x)-f(x')|, 
		\end{equation}
		which captures the magnitude by which a single entity’s data can change the output of $f$ in the worst case \cite{dwork2006differential}.
	\end{definition}

	\subsection{Additive Secret Sharing}
	\label{sec:ass}
	
	Given a private value $x\in \mathbb{Z}_{2^{k}}$, ASS in a two-party setting works by splitting it into two secret shares $\langle x\rangle_{1}$ and $\langle x\rangle_{2}$ such that $x = \langle x\rangle_{1}+\langle x\rangle_{2}$ \cite{mohassel2017secureml}.
	%
	Each share alone reveals no information about $x$.
	We denote by $\llbracket x \rrbracket$ the ASS of $x$ for short. 
	It is noted that if $k=1$, we say the secret sharing is \textit{binary sharing}, denoted as $\llbracket x \rrbracket^{B}$, and otherwise \textit{arithmetic sharing}, denoted as $\llbracket x \rrbracket^{A}$.
	%
	%
	%
	Given a public constant $c$, and the secret sharings $\llbracket x \rrbracket$ and $\llbracket y \rrbracket$, addition/subtraction $\llbracket x\pm y \rrbracket =\llbracket x \rrbracket \pm \llbracket y \rrbracket$ and scalar multiplication $\llbracket c\cdot x \rrbracket = c\cdot \llbracket x \rrbracket$ can be performed without interaction among the two parties that hold the shares respectively, while multiplication $\llbracket  x\cdot y \rrbracket=\llbracket  x \rrbracket \cdot\llbracket y\rrbracket$ requires the two parties to have one round of online communication with the use of Beaver triples which can be prepared offline.
	It is noted that addition and multiplication over $\mathbb{Z}_2$ are equivalent to XOR and AND respectively.

		\begin{algorithm}[!t]
		\caption{The Arnoldi Method} 
	
		\label{algo:1}
		\begin{algorithmic}[1] 
			\REQUIRE A non-symmetric matrix $\mathbf{A}$; the target dimension $M$.
			\ENSURE A new matrix $\overline{\mathbf{A}}$ with dimension $M*M$ and $\mathbf{P}$.
			\STATE Start with an arbitrary vector $\mathbf{p}_{1}$ with $L^{2}$ norm 1. \label{algo:ArnoStart}
			\FOR{$k \in [2,M]$}
			\STATE $\mathbf{p}_{k}=\mathbf{A}\mathbf{p}_{k-1}$.
			\FOR{$j \in [1,k-1]$}
			\STATE $\overline{\mathbf{A}}[j,k-1]=\mathbf{p}_{j}^{T}\cdot \mathbf{p}_{k}$. ~~~\# $T$ denotes transposition.
			\STATE $\mathbf{p}_{k}=\mathbf{p}_{k}-\overline{\mathbf{A}}[j,k-1]\cdot \mathbf{p}_{j}$.
			\ENDFOR \label{algo:ArnoEnd}
			\STATE $\overline{\mathbf{A}}[k,k-1]=||p_{k}||$. ~~~~~~~~~~\#$||\cdot||$ denotes $L^{2}$ norm. \label{alg:squareRoot}
			\STATE $\mathbf{p}_{k}=\frac{\mathbf{p}_{k}}{\overline{\mathbf{A}}[k,k-1]}$. \label{alg:division}
			\ENDFOR
			\STATE $\mathbf{P}=[\mathbf{p}_{1},\cdots,\mathbf{p}_{M}]$. \label{alg:AQ}
		\end{algorithmic}
		
	\end{algorithm}

	\subsection{Function Secret Sharing}
	\label{sec:fss}
	
	FSS \cite{boyle2015function} is a low-interaction secret sharing for secure computation, presenting prominent advantages in online communication and round complexity compared to other alternative techniques, such as garbled circuits \cite{yao1982protocols} or ASS.
	%
	 From a high-level point of view, a two-party FSS-based approach to a private function $f$ consists of a pair of probabilistic polynomial time (PPT) algorithms: (i) $(k_{1},k_{2})\leftarrow\mathsf{Gen}(1^{\lambda}, f)$: given a security parameter $\lambda$ and a function description $f$, output two succinct FSS keys $k_{1},k_{2}$, each for one party. (ii) $\langle f(x)\rangle_{i} \leftarrow\mathsf{Eval}(k_{i},x)$: given an FSS key $k_{i}$ and an evaluation point $x\in\{0,1\}^{n}$, output the share $\langle f(x)\rangle_{i}$ of the function evaluation result. The security of FSS ensures that an adversary learning only one FSS key $k_{i}, i\in\{1,2\}$ learns no information about the function $f$ and output $f(x)$.

		


	\section{System Overview}
	\label{sec:problem_def}
	
	\subsection{Architecture}

	%
	Fig. \ref{fig:systemmodel} illustrates \main's system architecture.
	%
	There are three kinds of entities: the users $\mathcal{U}_{i} (i\in[1,N]$), the cloud, and the analyst. 
	The users and the social relationship (friendship for simplicity) among them constitute a (decentralized) social graph, where each $\mathcal{U}_{i}$ represents a graph \textit{node} and the friendship between any two users indicates the existence of an \textit{edge}, and the number of each $\mathcal{U}_{i}$' associated friends is the \textit{degree} of the corresponding graph node.
	Consider as a concrete example a phone network \cite{qin2017generating}, where each $\mathcal{U}_{i}$ has a limited phone numbers in his contact list and thus a limited local view about the social graph.

	\begin{algorithm}[!t]
		\caption{The Lanczos Method} 
	
		\label{algo:2}
		\begin{algorithmic}[1] 
			\REQUIRE A symmetric matrix $\mathbf{A}$; the target dimension $M$.
			\ENSURE A small matrix $\overline{\mathbf{A}}$ with dimension $M*M$ and $\mathbf{P}$.
			\STATE Start with an arbitrary vector $\mathbf{p}_{1}$ with $L^{2}$ norm 1.
			\STATE $\mathbf{w}=\mathbf{A}\mathbf{p}_{1}$; $\alpha_{1}=\mathbf{w}^{T}\cdot \mathbf{p}_{1}$; $\mathbf{w}=\mathbf{w}-\alpha_{1}\cdot \mathbf{v}_{1}$.
			\FOR{$k \in [2,M]$}
			\STATE $\beta_{k}=||\mathbf{w}||$.
			\STATE $\mathbf{p}_{k}=\frac{\mathbf{w}}{\beta_{k}}$.
			\STATE $\mathbf{w}=\mathbf{A}\mathbf{p}_{k}$.
			\STATE$\alpha_{k}=\mathbf{w}^{T}\cdot \mathbf{p}_{k}$.
			\STATE$\mathbf{w}=\mathbf{w}-\alpha_{k}\cdot \mathbf{p}_{k}-\beta_{k}\cdot \mathbf{p}_{k-1}$.
			\ENDFOR
			\STATE All $\alpha$ and $\beta$ form the tridiagonal matrix $\overline{\mathbf{A}}$.
			\STATE $\mathbf{P}=[\mathbf{p}_{1},\cdots,\mathbf{p}_{M}]$. \label{alg:LQ}
		\end{algorithmic}

	\end{algorithm}

	The decentralized social graph can be characterized by an adjacency matrix $\mathbf{A}$ of size $N*N$, where each row $\mathbf{A}[i,:] (i\in[1,N]$) indicates $\mathcal{U}_{i}$'s local view. For example, in an unweighted social graph, $\mathbf{A}[i,j]=1$ may indicate that $\mathcal{U}_{i}$ and $\mathcal{U}_{j}$ are friends; in a weighted social graph, $\mathbf{A}[i,j]=v$ may indicate the degree of intimacy between $\mathcal{U}_{i}$ and $\mathcal{U}_{j}$ by a value $v$.
	%
	%
	The users are willing to allow analytics over their federated data of local views to produce analytical results for the analyst.
	Our focus in this paper is the analytics task of \eig, which plays a vital role in graph analytics and has many applications as aforementioned.
	%
	%
	However, due to privacy concerns, each $\mathcal{U}_{i}$ is not willing to disclose his private social relationship throughout the analytical process over the data federation, so the enforcement of data protection is demanded.

	The graph analytics service is empowered by the cloud for well understood benefits.
	In \main, the power of the cloud is split into two cloud servers from different trust domains which can be hosted by different service providers in practice.
	Such multi-server model is getting increasingly popular in recent years for security designs in various domains, including both academia \cite{du2020graphshield,chen2020metal,MohasselRR20, LiuZYY21,ZhengDTWZ21,xu2021privacy,tan2021cryptgpu,ZhengDW18} and industry \cite{Tfencrypt,Crypten}.
	The adoption of such model in \main~follows this trend.
	The two cloud servers in \main, denoted as {\csa}, collaboratively perform the analytics task of \eig~without seeing each $\mathcal{U}_i$'s data, and produce encrypted results (the top-$k$ {\eigs}) which are delivered to the analyst on demand. At a high level, {\main} proceeds through two phases: (i) \emph{Secure collection of decentralized social graph data}. In this phase, each user $\mathcal{U}_i$'s local view data $\mathbf{A}[i,:]$ regarding the social graph is collected by the cloud servers in encrypted form so as to form the matrix for \eig, through sparsity-aware and privacy-preserving techniques developed in \main. (ii) \emph{Secure {\eig}}. After the encrypted adjacency matrix is adequately formed at the cloud, the cloud servers collaboratively perform secure \eig, through a suite of customized secure protocols developed in \main.

		\begin{figure}[t!]
		\centering
		\includegraphics[width=0.8\linewidth]{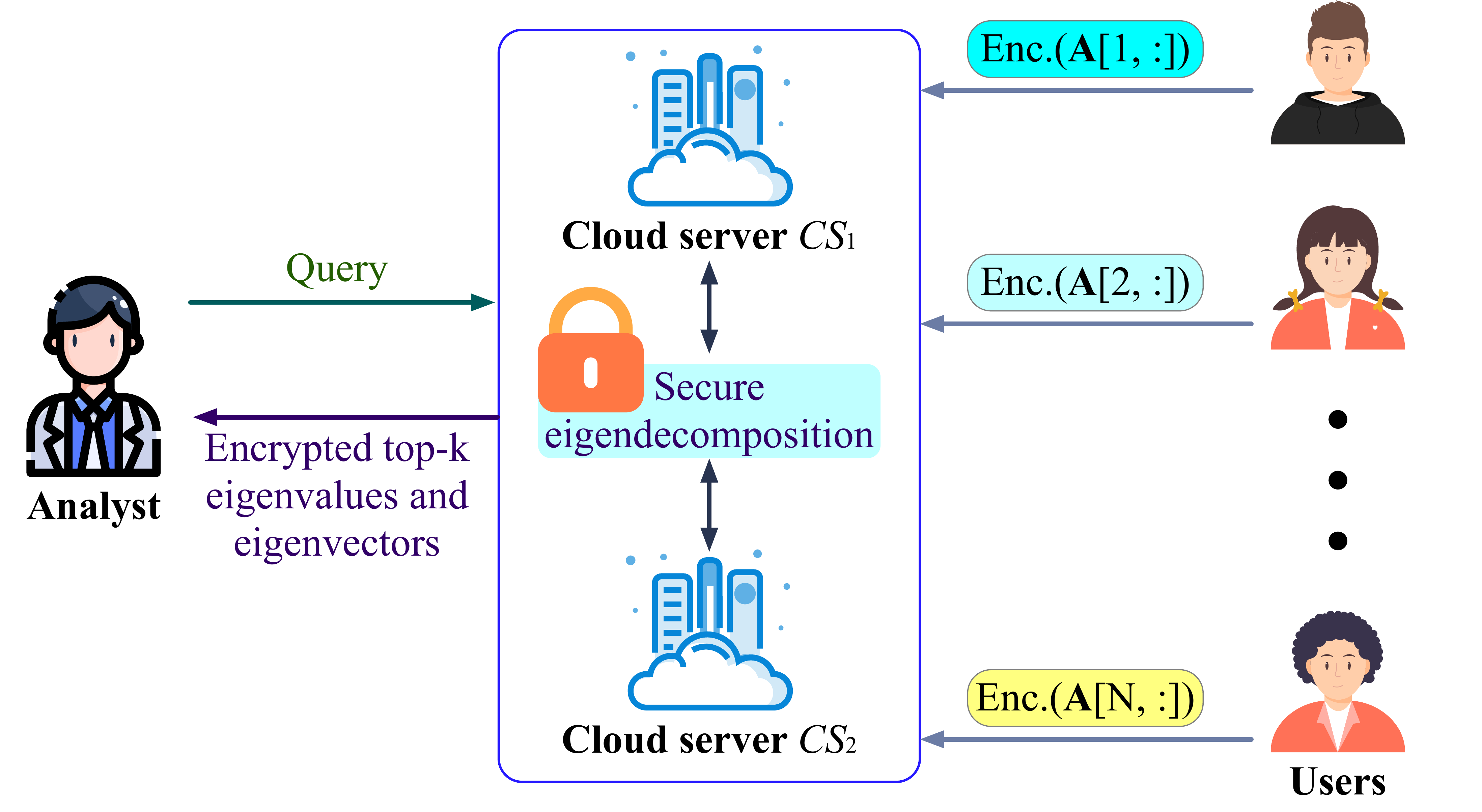}
		\caption{The system architecture of {\main}. $\mathbf{A}[i,:]$ represents the local view data from user $\mathcal{U}_{i}$ on a decentralized social graph.}
		\label{fig:systemmodel}
		\vspace{-15pt}
	\end{figure}

	\subsection{Threat Model}
	\label{sec:threat_model}
	
	Along the system workflow in \main, we consider that the primary threats come from the cloud entity empowering the \eig-based graph analytics service.
	Similar to most prior security designs in the two-server setting \cite{doerner2017scaling,agrawal2019quotient,chen2020metal}, we assume a semi-honest and non-colluding adversary model where {\csa} honestly follow the protocol specification of \main, but may individually attempt to learn private information about individual users' social connections in the decentralized social graph and the analytics result. 
	Note that each user's private social connections with other entities in the graph are reflected by the non-zero elements in her local view vector, as introduced above.
	So, for the privacy of individual users, {\main} aims to conceal the sensitive information regarding the non-zero elements in each user $\mathcal{U}_i$'s local view vector $\mathbf{A}[i,:]$, which includes the \textit{positions}, \textit{values}, and \textit{number}.

	\section{Secure Collection of Decentralized Social Graph Data}
	\label{sec:main}

	\subsection{Overview}
	
	In this phase, each user $\mathcal{U}_i$ provides his local view data $\mathbf{A}[i,:]$ in protected form to the service.
	For high efficiency, \main~resorts to the lightweight ASS technique for encryption of the elements in $\mathbf{A}[i,:]$.
	A simple method is to have each $\mathcal{U}_i$ directly apply ASS over his complete $\mathbf{A}[i,:]$ of length $N$.
	However, such simple method is clearly inefficient due to the sparsity of the (decentralized) social graph.
	According to Facebook's statistics \cite{curtiss2013unicorn}, on average a user has 130 friends in a social network, which is far less than $N$ (e.g., hundreds of thousands). 
	This leads to high sparsity, indicating that the complete $\mathbf{A}[i,:]$ will be filled with many zeros.
	So the simple method will incur significant cost on the user side as well as pose unnecessary workload in the subsequent secure \eig~process.
	
	To remedy this, a plausible approach is to leverage sparse representation. Specifically, $\mathcal{U}_{i}$ applies ASS only over each nonzero element at location $j$ (denoted as $\{(i,j,\mathbf{A}[i,j])\}$), and sends $\{(i,j,\llbracket\mathbf{A}[i,j]\rrbracket^{A})\}$ to {\csa}. 
	Such approach yet leads to prominent privacy leakages: (i) The number of nonzero elements (i.e., the degree) indicates the number of $\mathcal{U}_{i}$'s friends, which can be used in inferring $\mathcal{U}_{i}$'s privacy \cite{zhou2008brief}. (ii) The presence of a nonzero element $(i,j,\llbracket\mathbf{A}[i,j]\rrbracket^{A})$ implies the presence of the relationship between $\mathcal{U}_{i}$ and $\mathcal{U}_{j}$, which can also be used in inference attacks \cite{zhou2008brief}. (iii) If the social graph is unweighted (i.e., each element is 0 or 1), the presence of a nonzero element $(i,j,\llbracket\mathbf{A}[i,j]\rrbracket^{A})$ implies $\mathbf{A}[i,j]=1$, revealing the data to {\csa}.

	Therefore, the challenge here is how to preserve the benefits of sparsity while protecting the privacy of individual user's social relationships.
	Meanwhile, the effectiveness of subsequent \eig~should not be affected.

	Our key insight is to delicately trade off (\emph{node degree}) privacy for efficiency, by having $\mathcal{U}_i$ blend in some dummy edges with \emph{zero} weights at random empty locations in $\mathbf{A}[i,:]$, and apply ASS over the weights of both true and dummy edges, inspired by \cite{sharma2018privategraph}.
	Under ASS, even encrypting the same (zero) value multiple times will result in different shares (ciphertexts) indistinguishable from uniformly random values.
	%
	Therefore, the dummy edges cannot be distinguished from the true edges, as well as do not affect the effectiveness of the subsequent secure {\eig} process.
	%


	What remains challenging here is how to appropriately set the number of dummy edges so as to delicately balance the trade-off between \textit{efficiency} and \textit{privacy}.
	%
	%
	%
	Too many dummy edges will impair the sparsity and increase the system overhead, while too few dummy edges will result in weaker privacy protection.
	Therefore, a tailored security design is demanded to provide a theoretically sound approach by which $\mathcal{U}_{i}$ can select the adequate number of dummy edges to achieve a balance between efficiency and privacy. 
	Our main idea is to rely on LDP so as to make the leakage about the node degrees \emph{differentially private}.
	%
	In what follows, we start with some basic approaches and discuss their limitations.
	Then we present our tailored solution.
	
	\subsection{Basic Approaches}

		\begin{figure}
		\centering
		\includegraphics[width=0.65\linewidth]{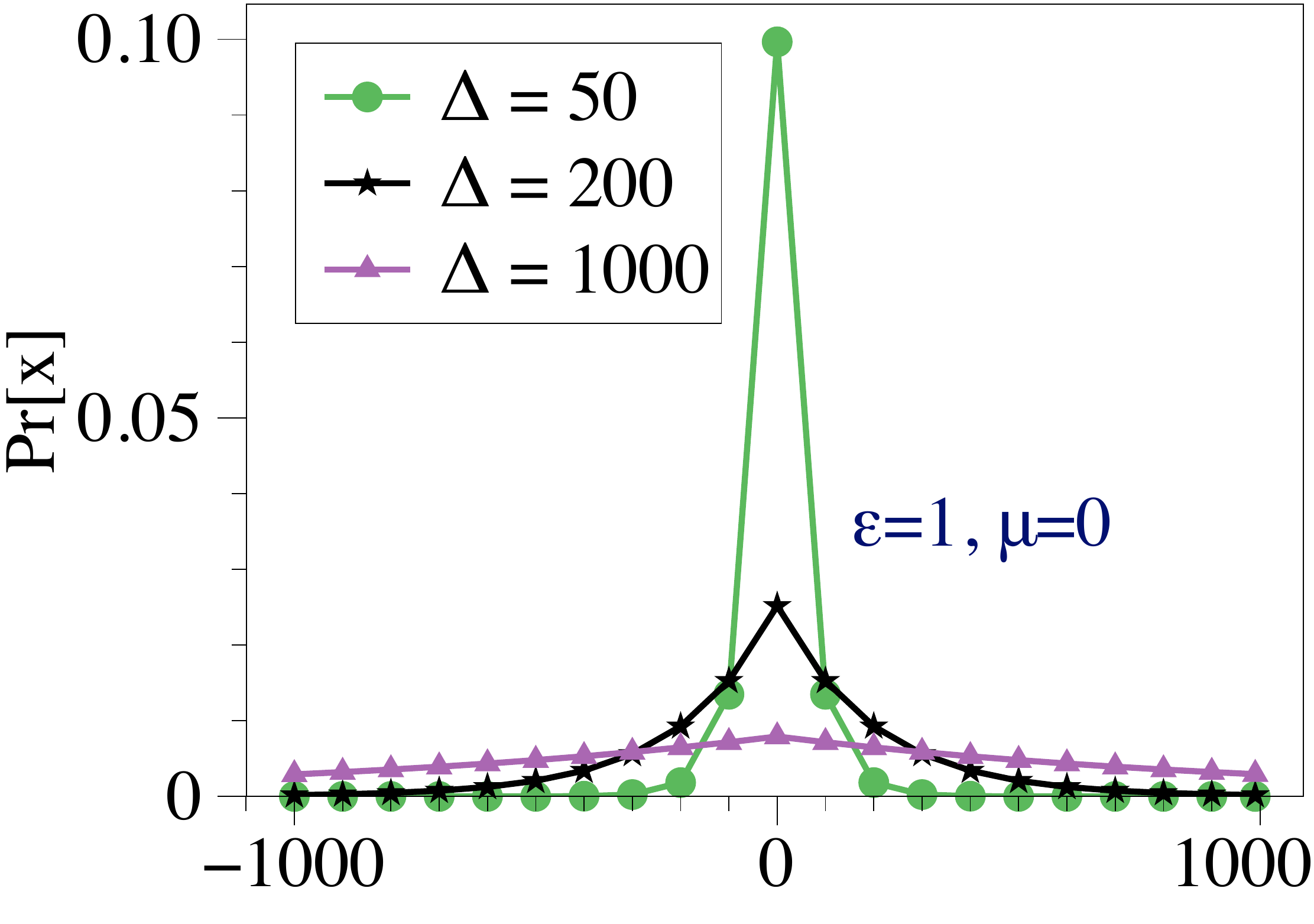}
		\caption{The discrete Laplace distribution with different sensitivities $\Delta$.}
		\label{fig:viewLap}
		\vspace{-12pt}
	\end{figure}
	
	A basic approach based on LDP can work as follows.
	As aforementioned, each $\mathcal{U}_{i}$ draws a noise $n_{i}$ from the discrete Laplace distribution, and then blends $n_{i}$ dummy edges with zero values at random empty locations in $\mathbf{A}[i,:]$, followed by encrypting the weights of both true and dummy edges by ASS.
	In this way, differential privacy guarantee on the node degree can be achieved, and the presence of edges is disguised as well \cite{sharma2018privategraph}.
	%
	%
	Here, the sensitivity $\Delta$ needs to be set to $\Delta=(d_{max}-d_{min})$, where $d_{max}$ and $d_{min}$ are the possible maximum and minimum degree in the decentralized social graph, respectively. 
	The above basic approach, however, would result in a very large $\Delta$ (up to $N$ in theory) to be applied, leading to a large number of dummy edges to be added and heavily impairing the sparsity. 
	To illustrate why this is the case, Fig. \ref{fig:viewLap} shows the probability density function of the discrete Laplace distribution with different $\Delta$. 
	It is revealed that a larger $\Delta$ leads to the shape of the density function being more uniform. 
	This indicates that the larger the sensitivity $\Delta$ is, the larger probability that $\mathcal{U}_{i}$ draws a large $|n_{i}|$ will be.
	In contrast, a small $\Delta$ will make the probability density function concentrated (e.g., $\Delta=50$ in Fig. \ref{fig:viewLap}), which means that $\mathcal{U}_{i}$ will draw a small $|n_{i}|$ with a large probability. 
	%
	%
	

	In order to achieve better efficiency, an alternative approach as proposed by \cite{sharma2018privategraph} is to use a bin-based mechanism that provides bin-wise differential privacy rather than graph-wise differential privacy.
	%
	The main idea is to partition users into several bins, each of which contains an approximately equal number of users whose degrees are within a small interval $[d_{p}, ~d_{q}]$. 
	As such, all users in the same bin can use a much smaller sensitivity $\Delta=d_{q}-d_{p}$.
	%
	We note that while the idea proposed by Sagar \textit{et al.} \cite{sharma2018privategraph} is useful, their design to instantiate such idea is not satisfactory and has crucial downsides. Firstly, it requires some (sampled) users to open their degrees to estimate the degree histogram of the social graph so as to split users into bins, harming the privacy of the sampled users.
	Secondly, to avoid edge deletion when the drawn noise is negative, it lets all users add a large offset to the drawn noise, making a notable impact on the sparsity.
	In addition, no formal analysis regarding the differential privacy guarantee is provided.

	\subsection{Our Approach}
	\label{sec:phase1}

	Based on the aforementioned idea of bin-based LDP, we design a new protocol for secure collection of decentralized social graph data, which does not expose any users' exact node degree information, as opposed to PrivateGraph \cite{sharma2018privategraph}.
	Meanwhile, our protocol does not require users to add a large offset to the noises, and thus is much more advantageous in maintaining the sparsity.
	We also later provide formal analysis on the differential privacy guarantee.
	Our protocol for secure collection of decentralized social graph data is comprised of the following ingredients: (i) secure degree histogram estimation, (ii) secure binning map generation, and (iii) local view data encryption.

	\subsubsection{Secure Degree Histogram Estimation}
	\label{sec:1-1}
	
	It is challenging for {\csa} to obliviously estimate the degree histogram without knowing any users' degrees. 
	In particular, given a public degree $d_{i}$ and a specific user's degree $d_{j}$, {\csa} need to efficiently check whether $d_{j}=d_{i}$ with both $d_{j}$ and the checking result being encrypted.
	To counter this challenge, we devise a tailored construction based on function secret sharing \cite{boyle2015function}, a recently developed cryptographic primitive that allows the generation of compact function shares and secure evaluation with the function shares.
	In particular, we take advantage of distributed point function (DPF) \cite{boyle2016function}, an FSS scheme for the point function $f_{\alpha,\beta}(x)$ which outputs $\beta$ if $x=\alpha$ and otherwise outputs 0. 
	Formally, a two-party DPF, parameterized by a finite Abelian group $\mathbb{G}$ ($\mathbb{G}:=\mathbb{Z}_{2^l}$ in our design), consists of the following algorithms: (i) $(k_{1},k_{2})\leftarrow\mathsf{Gen}(1^{\lambda}, \alpha,\beta)$: given a security parameter $1^{\lambda}$, a string $\alpha\in\{0,1\}^{n}$, and a value $\beta\in \mathbb{G}$, outputs two succinct DPF keys representing the function shares. (ii) $\langle f_{\alpha,\beta}(x)\rangle_{i} \leftarrow\mathsf{Eval}(k_{i},x)$: given a DPF key $k_{i}$ and an input $x\in\{0,1\}^{n}$, outputs $\langle f_{\alpha,\beta}(x)\rangle_{i}\in \mathbb{G}$.

	\begin{algorithm}[!t]
		\caption{Secure Degree Histogram Estimation} 
	
		\label{algo:3}
		\begin{algorithmic}[1] 
			\REQUIRE Users $\mathcal{SU}_{j}$'s degree $d_{j}$, $j\in[1,S]$.
			\ENSURE The encrypted estimated degree histogram $\llbracket D_{i}\rrbracket^{A}$, for all possible degrees $i\in[1,d_{max}]$.
			
			\underline{\# At $\mathcal{SU}_{j}, j\in [1,S]$ side:}
			\STATE Generate $(k_{j,1},k_{j,2})\leftarrow\mathsf{Gen}(d_{j}, 1)$. \label{alg:genDpf}
			\STATE Send $k_{j,1},k_{j,2}$ to {\cs}, {\css}, respectively.
			
			\underline{\# At $\mathcal{CS}_{t},t\in\{1,2\}$ side:}
			\STATE Initialize $\langle D_{i}\rangle_{t}=0,i\in[1,d_{max}]$.
			\FOR{$i\in[1,d_{max}]$}
			\FOR{$j\in[1,S]$}
			\STATE $\langle D_{i}\rangle_{t}+=\mathsf{Eval}(k_{j,t},i)$. \label{alg:evalDpf}
			\ENDFOR
			\ENDFOR
		\end{algorithmic}
	\end{algorithm}

	Building on top of DPF, we develop a technique for secure degree histogram estimation, as presented in Algorithm \ref{algo:3}.
	%
	%
	As directly engaging all users for building the degree histogram will lead to high performance overhead, \main~adopts a sampling strategy following \cite{sharma2018privategraph}, where $S$ users, denoted by $\{\mathcal{SU}_{j}\}_{j\in[1,S]}$, are randomly sampled to participate.
	%
	%
	%
	Each $\mathcal{SU}_{j}$ generates DPF keys based on his degree $d_{j}$ (line \ref{alg:genDpf} in Algorithm \ref{algo:3}), where $\alpha$ and $\beta$ are set to $d_{j}$ and $1$ respectively.
	%
	%
	Each $\mathcal{SU}_{j}$ then sends $k_{j,1}$ to $\mathcal{CS}_1$ and $k_{j,2}$ to $\mathcal{CS}_2$.
	Utilizing the DPF keys $\{k_{j,t}\}_{j\in[1,S]}$, each $\mathcal{CS}_{t}$ ($t\in\{1,2\}$) can evaluate $\mathsf{Eval}(k_{j,t},i)$ for all possible degrees $i\in[1,d_{max}]$. By summing these evaluation results (line \ref{alg:evalDpf} in Algorithm \ref{algo:3}), $\mathcal{CS}_{t}$ can obtain exactly the encrypted number $\langle D_{i}\rangle_{t}$ of sampled users whose degree is equal to $i$. Correctness holds since 
	\vspace{-8pt}
	\begin{align}\notag
		\langle D_{i}\rangle_{1}+	\langle D_{i}\rangle_{2}&=\sum_{j=1}^{S}\mathsf{Eval}(k_{j,1},i)+\sum_{j=1}^{S}\mathsf{Eval}(k_{j,2},i)\\\notag
		&=\sum_{j=1}^{S}[\mathsf{Eval}(k_{j,1},i)+\mathsf{Eval}(k_{j,2},i)]\\\notag
		&=\sum_{j=1}^{S}1\{d_{j}=i\}.
	\end{align}
	
	
	\subsubsection{Secure Binning Map Generation}
	\label{sec:1-2}
	
	With the encrypted degree histogram produced on the cloud side, we now introduce how to enable {\csa} to generate an encrypted binning map, based on which each user can get the sensitivity value for his use in drawing the noise.
	%
	%
	The binning strategy underlying our secure design follows prior art \cite{sharma2018privategraph}, leading to each bin containing an approximately equal number of users whose degrees are within a small interval.
	%
	%
	It is noted that unlike our design operating in the ciphertext domain, \cite{sharma2018privategraph} operates in \emph{plaintext domain} with \emph{exposed} raw degrees of sampled users for building the histogram and producing the binning map.
	%


	
	Our secure design is shown in Algorithm \ref{algo:4}, which inputs the encrypted degree histogram $\llbracket D_{i}\rrbracket^{A},i\in[1,d_{max}]$, the number $B$ of bins, and the number $S$ of the sampled users, and then outputs the encrypted binning map $\llbracket\mathsf{Inter}\rrbracket^{B}\in\{\llbracket 0\rrbracket^{B},\llbracket 1\rrbracket^{B}\}^{d_{max}}$ (i.e., an encrypted bit-string under binary secret sharing), where $1$ indicates the boundary of a bin. For example, given $d_{max}=10$, $\mathsf{Inter}=0001000001$ indicates that the users are partitioned into two bins: users whose degree $\in[1,4]$ and users whose degree $\in[5,10]$. 
	After obliviously computing $\llbracket\mathsf{Inter}\rrbracket^{B}$, {\csa} send it to all users, and then each of them can judge to which bin he belongs based on his own degree locally. 
	
	As shown in Algorithm \ref{algo:4}, {\csa} first calculate the public bin size $sizeB$ based on $B$ and $S$ (line \ref{alg:sizeb}), and then initialize an accumulator $\llbracket accu\rrbracket^{A}$. After that, {\csa} add $\llbracket D_{i}\rrbracket^{A}$ ($i\in [1,d_{max}]$) to $\llbracket accu\rrbracket^{A}$ in turn (line \ref{alg:accu}). After each addition, {\csa} obliviously evaluate $\llbracket\llbracket accu\rrbracket^{A}\ge sizeB\rrbracket^{B}$ and add the comparison result (in binary secret sharing) to $\llbracket\mathsf{Inter}[i]\rrbracket^{B}$. Specifically, if $ accu \ge sizeB$, $\mathsf{Inter}[i]=1$, indicating a bin boundary; otherwise, $\mathsf{Inter}[i]=0$. After that, based on $\llbracket\mathsf{Inter}[i]\rrbracket^{B}$, {\csa} obliviously evaluate whether to reset the accumulator $\llbracket accu\rrbracket^{A}$ or not. Specifically, if $\mathsf{Inter}[i]=1$, a bin boundary appears, and thus $accu$ needs to be reset to 0 for the next bin; otherwise, $accu$ remains unchanged. This step is given in line \ref{alg:reset}, where ``$!$" denotes \textit{not} operation which can be achieved by letting one of {\csa} flip its share $\langle\mathsf{Inter}[i]\rangle_{1}$ or $\langle\mathsf{Inter}[i]\rangle_{2}$ locally. Finally, {\csa} output $\llbracket\mathsf{Inter}\rrbracket^{B}\in\{\llbracket 0\rrbracket^{B},\llbracket 1\rrbracket^{B}\}^{d_{max}}$.

	\begin{algorithm}[!t]
		\caption{Secure Binning Map Generation} 
		
		\label{algo:4}
		\begin{algorithmic}[1] 
			\REQUIRE The encrypted estimated degree histogram $\llbracket D_{i}\rrbracket^{A},i\in[1,d_{max}]$, the number of bins $B$, and the number of sampled users $S$.
			\ENSURE The encrypted binning map $\llbracket\mathsf{Inter}\rrbracket^{B}\in\{\llbracket 0\rrbracket^{B},\llbracket 1\rrbracket^{B}\}^{d_{max}}$. 
			\STATE Calculate the bin size $sizeB=\frac{S}{B}$. \label{alg:sizeb}
			\STATE Initialize the accumulator $\llbracket accu\rrbracket^{A}=0$.
			\FOR{$i\in[1,d_{max}]$}
			\STATE $\llbracket accu\rrbracket^{A}+=\llbracket D_{i}\rrbracket^{A}$. \label{alg:accu}
			\STATE $\llbracket\mathsf{Inter}[i]\rrbracket^{B}=\llbracket\llbracket accu\rrbracket^{A}\ge sizeB\rrbracket^{B}$. \label{alg:inter}
			\STATE $\llbracket accu\rrbracket^{A}=\llbracket !\mathsf{Inter}[i]\rrbracket^{B}\cdot\llbracket accu\rrbracket^{A} $. \label{alg:reset}
			\ENDFOR
		\end{algorithmic}
	\end{algorithm}
	
	In Algorithm \ref{algo:4}, the addition operations are easily supported with secret sharing, but the operation $\llbracket\llbracket accu\rrbracket^{A}\ge sizeB\rrbracket^{B}$ is not directly supported. 
	What we need here essentially is secure comparison in the secret sharing domain.
	From the very recent works \cite{boyle2021function,LiuZYY21}, we identify two primitives suited to allow secure comparison in the secret sharing domain, and introduce two approaches accordingly to allow the realization of secure binning map generation.
	%
	%
	The first approach is based on FSS \cite{boyle2021function}, which is more suited for high-latency network scenarios because it requires minimal rounds of interactions (at the cost of more local computation).
	The second approach is based on ASS \cite{LiuZYY21}, which requires a small amount of local computation but higher online communication and round complexities, fitting better into low-latency network scenarios.

	\noindent\textbf{FSS-based approach.} We identify the state-of-the-art construction of FSS-based secure comparison from \cite{boyle2021function}, which is referred to as distributed comparison function (DCF).
	%
	DCF is an FSS scheme for the function $g_{\alpha,\beta}(x)$ which outputs $\beta$ if $x<\alpha$ and $0$ otherwise. 
	%
	Formally, a two-party DCF, parameterized by two finite Abelian groups $\mathbb{G}^{in}, \mathbb{G}^{out}$, consists of the following algorithms: (i) $(k_{1},k_{2}, r_{1},r_{2})\leftarrow\mathsf{Gen}(1^{\lambda}, \alpha,\beta)$, which takes as input a security parameter $1^{\lambda}$, $\alpha\in\mathbb{G}^{in}$, and $\beta\in\mathbb{G}^{out}$, and outputs two keys $k_{1},k_{2}$ and two random values $r_{1},r_{2}\in\mathbb{G}^{in}$ ($r_{1}+r_{2}=r^{in}$), each for one party. (ii) $\langle g_{\alpha,\beta}(x)\rangle_{i} \leftarrow\mathsf{Eval}(k_{i},x+r^{in})$, which takes as input a key $k_{i}$ and a (masked) input $x+r^{in}\in \mathbb{G}^{in}$, and outputs $\langle g_{\alpha,\beta}(x)\rangle_{i}\in \mathbb{G}^{out}$.
	The evaluation process in DCF only requires one round of online communication, in which the two parties send $\langle x\rangle_{i}+r_{i},i\in\{1,2\}$ to each other to reveal $x+r^{in}$ without leaking $x$. The security of DCF states that if an adversary only learns one of $k_{i},r_{i},i\in\{1,2\}$, it learns no information about the private input $x$ and output $g_{\alpha,\beta}(x)$.

	We now show how \main~builds on the DCF for secure evaluation of $\llbracket\llbracket accu\rrbracket^{A}\ge sizeB\rrbracket^{B}$.
	To use the DCF in \main, the related public parameters can be set as follows: $\alpha=sizeB$, $\beta=1$, $\mathbb{G}^{in}=\mathbb{Z}_{2^{l}}$, and $\mathbb{G}^{out}=\mathbb{Z}_{2}$.
	With these parameters, the DCF keys $k_{1},k_{2}$ and the random values $r_{1},r_{2}\in\mathbb{Z}_{2^{l}}$ can be prepared offline and distributed to {\csa} respectively.
	Note that such offline work can be done by a third-party server in practice \cite{boyle2021function}.
	For the secure online evaluation of $\llbracket\llbracket accu\rrbracket^{A}\ge sizeB\rrbracket^{B}$, $\mathcal{CS}_{t}$ ($t\in\{1,2\}$) first exchanges $\langle accu\rangle_{t}+r_{t}$ to each other to reveal $accu+r^{in}$, and then evaluates $\mathsf{Eval}(k_{t},accu+r^{in})$, which will output $\langle 1\rangle_{t}$ if $accu < sizeB$ and $\langle 0\rangle_{t}$ otherwise. 
	However, {\main} requires {\csa} to output $\llbracket 1\rrbracket^{B}$ if $accu\ge sizeB$, and thus {\main} further lets one of {\csa} flip its share locally.

	\begin{figure}[!t]
		\centering
		\includegraphics[width=0.7\linewidth]{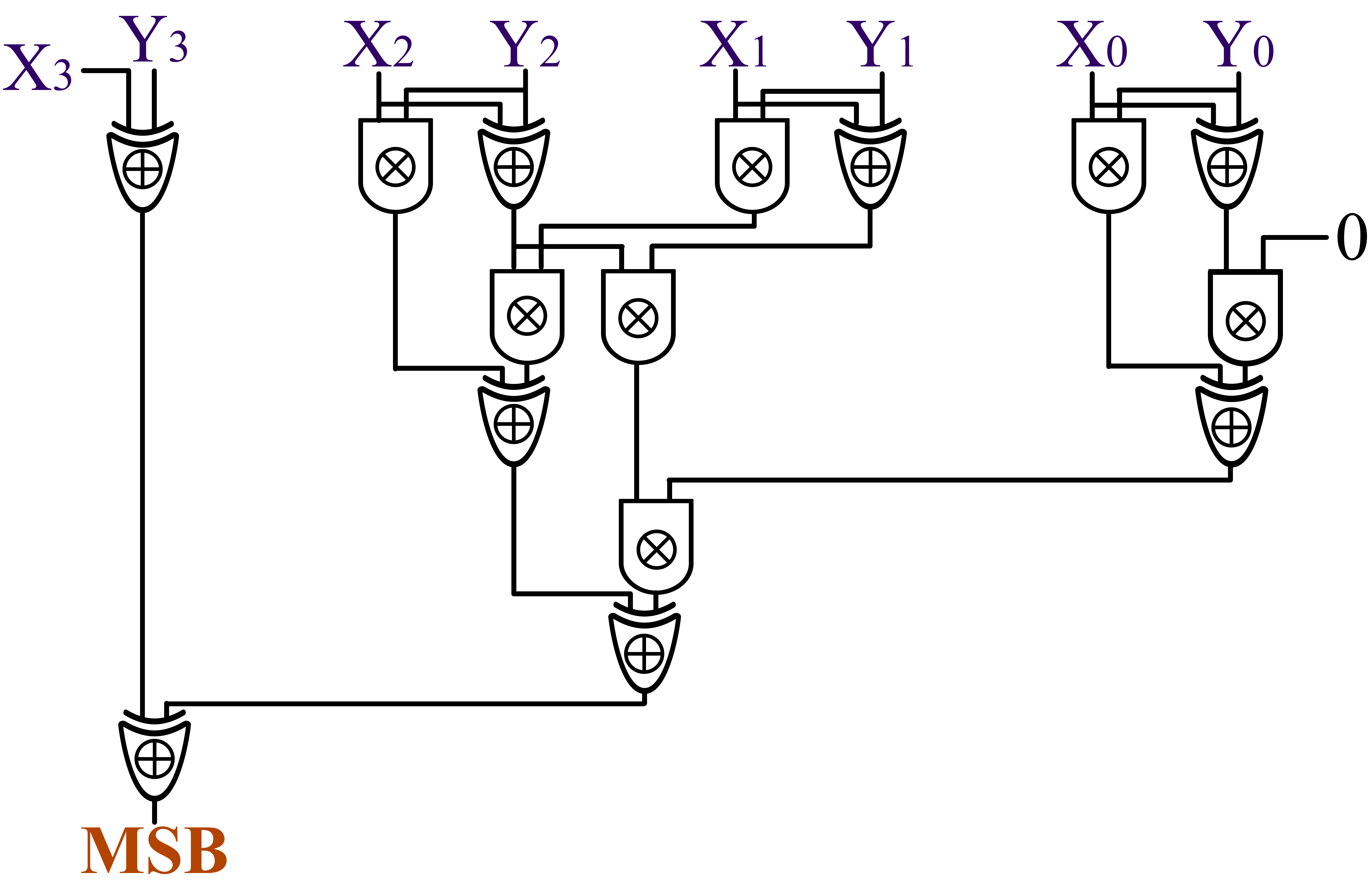}
		\caption{Illustration of a 4-bit tailored PPA.}
		\label{fig:ppa}
	\vspace{-10pt}
	\end{figure}

	\noindent\textbf{ASS-based approach.} 
	This approach is based on the idea of secure bit decomposition in the secret sharing domain \cite{LiuZYY21}.
	%
	Specifically, given two fixed-point integers $X, Y\in\mathbb{Z}_{2^k}$ under two's complement representation, the most significant bit (MSB) of $X-Y$ can indicate the relationship between $X-Y$ and 0. Namely, if $X-Y<0$, $msb(X-Y)=1$ and otherwise $msb(X-Y)=0$. 
	Secure extraction of the MSB in the secret sharing domain can be achieved via secure realization of a parallel prefix adder (PPA) \cite{LiuZYY21}, as illustrated in Fig. \ref{fig:ppa}. It is observed that only XOR and AND operations need to be performed in the secret sharing domain.
	\main~leverages the construction from \cite{LiuZYY21}, which takes the secret sharings of two values $X$ and $Y$ as input, and outputs $\llbracket msb(X-Y)\rrbracket^{B}$.
	With that construction, {\csa} are able to obtain $\llbracket msb(accu-sizeB)\rrbracket^{B}$, i.e., $\llbracket 0 \rrbracket^{B}$ if $accu\ge sizeB$ and $\llbracket 1 \rrbracket^{B}$ otherwise.
	However, what is needed in {\main} instead is that {\csa} output $\llbracket 1\rrbracket^{B}$ if $accu\ge sizeB$ and $\llbracket 0\rrbracket^{B}$ otherwise.
	\main~further lets one of {\csa} flip its share $\langle msb(\sigma)\rangle_{t}$ ($t\in \{1,2\}$) locally.

It is noted that the output $\llbracket\mathsf{Inter}[i]\rrbracket^{B}=\llbracket\llbracket accu\rrbracket^{A}\ge sizeB\rrbracket^{B}$ of the above two approaches are in binary secret sharing. 
However, for the computation in line \ref{alg:reset} of Algorithm \ref{algo:4}, the binary secret-shared value needs to be multiplied by an arithmetic secret-shared value, i.e., $\llbracket\mathsf{!Inter}[i]\rrbracket^{B}\cdot \llbracket accu\rrbracket^{A}$.
We note that this can be achieved based on an existing technique from \cite{mohassel2018aby3}, through the following steps:
%
(i) {\cs} randomly generates $r\in\mathbb{Z}_{2^{k}}$ ($k$ is the value length of $accu$), and sends two messages to {\css}: $m_{b}:=(b\oplus \langle  !\mathsf{Inter}[i] \rangle_{1})\cdot\langle accu\rangle_{1}-r, b\in\{0,1\}$. (ii) {\css} chooses $m_{0}$ if $  \langle  !\mathsf{Inter}[i]  \rangle_{2}=0$, otherwise {\css} chooses $m_{1}$. Therefore, the share held by {\css} is $m_{ \langle  !\mathsf{Inter}[i]  \rangle_{2}}= !\mathsf{Inter}[i] \cdot\langle accu\rangle_{1}-r$ and the share held by {\cs} is $r$.
(iii) For the other secret share $\langle accu\rangle_{2}$, {\css} acts as the sender and {\cs} acts as the receiver to perform step 1), 2) again.
At the end of the process, $\mathcal{CS}_{t}$ holds $\langle !\mathsf{Inter}[i]\cdot accu\rangle_{t}, t\in\{1,2\}$.
%
%
The resulting secret sharing of the binning map $\llbracket\mathsf{Inter}\rrbracket^{B}\in\{\llbracket 0\rrbracket^{B},\llbracket 1\rrbracket^{B}\}^{d_{max}}$ can then be sent to all users for use.

		\subsubsection{Local View Data Encryption}
		\label{sec:1-3}

		With the binning map $\mathsf{Inter}$ recovered at the user side, each $\mathcal{U}_i$ then encrypts his local view data, as shown in Algorithm \ref{algo:5}.
		The main problem is that the drawn noise $n_{i}$ could be negative (recall Fig. \ref{fig:viewLap}), which means that $\mathcal{U}_{i}$ should delete some edges. Obliviously, this will seriously impair the accuracy of subsequent secure {\eig}.
		To tackle this problem, {\main} lets $\mathcal{U}_{i}$ truncate $n_{i}$ to 0 inspired by \cite{he2017composing}, i.e., $n'_{i}=max(n_{i},0)$. In Section \ref{sec:privAnaly}, we will formally prove that all users' degrees still satisfy $(\epsilon,\delta)$-LDP although they may truncate $n_{i}$ to 0.
		%
		%
		Let $\mathcal{L}_i$ denote the set of locations of both true edges and dummy edges, and $\{(i,j,\mathbf{A}[i,j])\}_{j\in \mathcal{L}_i}$ the resulting data of $\mathcal{U}_i$ after adding dummy edges with zero weights.
		Finally, $\mathcal{U}_{i}$ applies ASS over each $\mathbf{A}[i,j]$ and sends $\{(i,j,\llbracket \mathbf{A}[i,j]\rrbracket^{A})\}_{j\in\mathcal{L}_i}$ to the cloud.

		\begin{algorithm}[!t]
			\caption{Local View Data Encryption} 
			
			\label{algo:5}
			\begin{algorithmic}[1] 
				\REQUIRE The binning map $\mathsf{Inter}$; $\mathcal{U}_{i}$'s local view data $\{(i,j,\mathbf{A}[i,j])\}$ and degree $d_{i}$. 
				\ENSURE $\mathcal{U}_{i}$'s encrypted data $\{(i,j,\llbracket \mathbf{A}[i,j]\rrbracket^{A})\}_{j\in\mathcal{L}_i}$.
				\STATE $\mathcal{U}_{i}$ determines to which bin $\mathcal{B}_{j}$ he belongs based on $d_{i}$.
				\STATE According to the bin interval $\mathcal{B}_{j}=[L_{j},U_{j}]$, $\mathcal{U}_{i}$ can locally derive the sensitivity value $\Delta_{j}=U_{j}-L_{j}$.
				\STATE $\mathcal{U}_{i}$ draws $n_{i}$ from $Lap(\epsilon,\delta,	\Delta_{j})$.
				\STATE $n'_{i}=\max(n_{i},0)$.
				
				\STATE $\mathcal{U}_{i}$ blends $n'_{i}$ dummy edges with weights 0 in his local view data at random empty locations, and the resulting data is denoted as $\{(i,j,\mathbf{A}[i,j])\}_{j\in \mathcal{L}_i}$.
				
				\STATE $\mathcal{U}_{i}$ applies ASS over each $\mathbf{A}[i,j],j\in \mathcal{L}_i$ to obtain the ciphertext $\{(i,j,\llbracket \mathbf{A}[i,j]\rrbracket^{A})\}_{j\in\mathcal{L}_i}$.
			\end{algorithmic}
		\end{algorithm}

		\section{Secure Eigendecomposition}
		\label{sec:secure_eig}
		
		\subsection{Overview}
		
		The encrypted data collected from the users forms the encrypted adjacency matrix under sparse encoding for \eig~at the cloud.
		We denote it by $\llbracket\mathbf{A}^{*}\rrbracket^{A}:=\{(i,j,\llbracket \mathbf{A}[i,j]\rrbracket^{A})\}_{i\in[1,N],j\in\mathcal{L}^*}$, where $\mathcal{L}^*$ is a multi-set that contains the all locations from $\mathcal{L}_1$, $\mathcal{L}_2$, $\cdots$, $\mathcal{L}_N$.
		Upon deriving $\llbracket\mathbf{A}^{*}\rrbracket^{A}$, \main~needs to enable {\csa} to obliviously perform {\eig} on $\llbracket \mathbf{A}^{*}\rrbracket^{A}$, producing the encrypted top-$k$ {\eigs}.

		Firstly, {\main} introduces techniques (Section \ref{sec:dimRec}) to enable {\csa} to obliviously transform $\llbracket \mathbf{A}^{*}\rrbracket^{A}$ to $\llbracket \mathbf{\overline{A}}\rrbracket^{A}$ (in dense encoding) of a much smaller size $M*M$, with respect to the Arnoldi method (for general (possibly
		non-symmetric) matrices) and Lanczos method (for symmetric matrices) as introduced before in Section \ref{sec:preSpectral}.
		Then, {\main} provides secure realizations of the QR algorithm so as to achieve secure \eig~over the $\llbracket \mathbf{\overline{A}}\rrbracket^{A}$ to produce encrypted top-$k$ \eigs. 
		We give a basic design (Section \ref{sec:QR}) for the secure QR algorithm as a starting point, followed by a delicate optimized design (Section \ref{sec:vectorMaxtrix}) that further achieves a performance boost.

		\subsection{Secure Matrix Dimension Reduction}
		\label{sec:dimRec}

		We first consider how to enable {\csa} to obliviously execute the Arnoldi method on $\llbracket \mathbf{A}^{*}\rrbracket^{A}$ to output $\llbracket \mathbf{\overline{A}}\rrbracket^{A}$.
		Looking into the operations in Algorithm \ref{algo:1}, we note that the operations in line \ref{algo:ArnoStart}-\ref{algo:ArnoEnd}, which are comprised of addition and multiplication, can be naturally and securely realized in the arithmetic secret sharing domain by {\csa}. 
		%
		However, it remains challenging for {\csa} to obliviously perform the operations in lines \ref{alg:squareRoot}, \ref{alg:division}, because \textit{square root} and \textit{division} are not directly supported in the secret sharing domain.

		Our solution is to approximate the square root and division operations using basic operations (i.e., $+,\times$), so that they can be securely supported in the secret sharing domain. For secure square root $\llbracket\sqrt{x}\rrbracket^{A}$, inspired by the very recent work \cite{knott2020crypten}, {\main} utilizes a roundabout strategy to approximate the inverse square root $ \frac{1}{\sqrt{x}}$ by the iterative Newton-Raphson algorithm \cite{akram2015newton}:
		\begin{equation}
			\label{eq:square_root}
			y_{n+1}=\frac{1}{2}y_{n}(3-xy_{n}^{2}), 
		\end{equation}
		which will converge to $y_{n}\approx \frac{1}{\sqrt{x}}$. Obviously, both subtraction and multiplication operations are naturally supported in the secret sharing domain. After that, {\main} lets {\csa} multiply $\llbracket y_{n}\rrbracket^{A}$ by $\llbracket x\rrbracket^{A}$ to derive $\llbracket\sqrt{x}\rrbracket^{A}$. 
		For secure division $\llbracket\frac{ y}{x}\rrbracket^{A}$ in the secret sharing domain, we note that the main challenge is to compute the reciprocal $\llbracket\frac{1}{x}\rrbracket^{A}$. However, we also observe that the reciprocal of division in Algorithm \ref{algo:1} is  $\frac{1}{\overline{\mathbf{A}}[k,k-1]}$, which is exactly the inverse square root computed in line \ref{alg:squareRoot}. Therefore, {\main} lets {\csa} directly perform the operation in line \ref{alg:division} by multiplying the  inverse square root computed in line \ref{alg:squareRoot} by $\llbracket \mathbf{q}_{k}\rrbracket^{A}$.
		
		\begin{algorithm}[!t]
			\caption{Secure Arnoldi Method} 
			\label{algo:7}
			\begin{algorithmic}[1] 
				\REQUIRE $\llbracket \mathbf{A}^{*}\rrbracket^{A}$ and the target dimension $M$.
				\ENSURE $\llbracket\overline{\mathbf{A}}\rrbracket^{A}$ with dimension $M*M$ and $\llbracket\mathbf{P}\rrbracket^{A}$.
				\STATE Start with an arbitrary vector $\llbracket\mathbf{ p}_{1}\rrbracket^{A}$ with $L^{2}$ norm 1. 
				\FOR{$k \in [2,M]$}
				\STATE $\llbracket \mathbf{p}_{k}\rrbracket^{A}=\llbracket \mathbf{A}^{*}\rrbracket^{A}\llbracket \mathbf{p}_{k-1}\rrbracket^{A}$.
				\FOR{$j \in [1,k-1]$}
				\STATE $\llbracket \overline{\mathbf{A}}[j,k-1]\rrbracket^{A}=\llbracket \mathbf{p}_{j}^{T}\rrbracket^{A}\cdot \llbracket \mathbf{p}_{k}\rrbracket^{A}$. 
				\STATE $\llbracket \mathbf{p}_{k}\rrbracket^{A}=\llbracket \mathbf{p}_{k}\rrbracket^{A}-\llbracket\overline{\mathbf{A}}[j,k-1]\rrbracket^{A}\cdot \llbracket \mathbf{p}_{j}\rrbracket^{A}$.
				\ENDFOR 
				\STATE $\llbracket x\rrbracket^{A}=\llbracket \mathbf{p}_{k}^{T}\rrbracket^{A}\cdot\llbracket \mathbf{p}_{k}\rrbracket^{A}$.
				
				\# Calculate the inverse square root of Eq. \ref{eq:cs}:
				\FOR{$n\in[1,\Omega]$} 
				\STATE $\llbracket y_{n+1} \rrbracket^{A} =\frac{1}{2}\cdot \llbracket y_{n}\rrbracket^{A}\cdot(3-\llbracket x\rrbracket^{A}\cdot\llbracket y_{n}\rrbracket^{A}\cdot\llbracket y_{n}\rrbracket^{A})$.
				\ENDFOR 
				
				\STATE $\llbracket\overline{\mathbf{A}}[k,k-1]\rrbracket^{A}=\llbracket x\rrbracket^{A}\cdot\llbracket y_{\Omega+1} \rrbracket^{A}$. 
				\STATE $\llbracket \mathbf{p}_{k}\rrbracket^{A}=\llbracket y_{\Omega+1}\rrbracket^{A} \cdot\llbracket \mathbf{p}_{k}\rrbracket^{A}$. 
				\ENDFOR
				\STATE $\llbracket\mathbf{P}\rrbracket^{A}=[\llbracket \mathbf{p}_{1}\rrbracket^{A},\cdots,\llbracket \mathbf{p}_{M}\rrbracket^{A}]$. \label{alg:P}
			\end{algorithmic}
		\end{algorithm}

		We give our construction for securely realizing the Arnoldi method in Algorithm \ref{algo:7}.
		Regarding the Lanczos method, we note that the operations required to be securely supported are identical to the Arnoldi method, so we omit the algorithm description for the secure Lanczos method.
		It is noted that the encrypted adjacency matrix $\llbracket\mathbf{A}^{*}\rrbracket^{A}:=\{(i,j,\llbracket \mathbf{A}[i,j]\rrbracket^{A})\}_{i\in[1,N],j\in\mathcal{L}^*}$ under sparse encoding can significantly save the cloud-side cost. For example, the secure matrix-vector product between $\llbracket\mathbf{A}^{*}\rrbracket^{A}$ and $\llbracket\mathbf{v}\rrbracket^{A}$ can be efficiently performed by only securely multiplying the elements $\llbracket \mathbf{A}[i,j]\rrbracket^{A}$ and the corresponding $\llbracket\mathbf{v}[j]\rrbracket^{A}$. 
		Obliviously, the number of multiplications in this example is independent of the number of columns in the original complete $\mathbf{A}$ and only depends on the number of rows and the number of nonzero elements in $\mathbf{A}$.

		\subsection{Secure QR Algorithm}
		\label{sec:QR}
		
		We now introduce, as a basic design, how to enable {\csa} to obliviously calculate the complete encrypted {\eigs} of $\llbracket\overline{ \mathbf{A}}\rrbracket^{A}$, through a protocol for securely realizing the widely used QR algorithm.
		The QR algorithm works in an iterative manner, consisting of a series of QR decomposition. Formally, given matrix $\mathbf{L}$ and $\mathbf{T}_{0}=\mathbf{L}$, in the $k$-th ($k\in[1,K]$) iteration, on input $\mathbf{T}_{k-1}$, we compute QR decomposition $\mathbf{T}_{k-1}=\mathbf{Q}_{k-1}\mathbf{R}_{k-1}$, where $\mathbf{Q}_{k-1}$ is an orthogonal matrix (i.e., $\mathbf{Q}^{T}=\mathbf{Q}^{-1}$), $\mathbf{R}_{k-1}$ is an upper Hessenberg matrix, and then output $\mathbf{T}_{k}=\mathbf{R}_{k-1}\mathbf{Q}_{k-1}$. At the end of QR algorithm, the diagonal elements of $\mathbf{T}_{K}$ are $\mathbf{L}$'s eigenvalues and $\mathbf{S}=\mathbf{Q}_{1}\cdots\mathbf{Q}_{K}$ are $\mathbf{L}$'s eigenvectors.

		In \main, we perform the QR decomposition utilizing the Givens rotations \cite{press1987numerical}. 
		%
		%
		Formally, given an $M*M$ upper Hessenberg matrix $\mathbf{T}_{k-1}$,
		we create the orthogonal Givens rotation matrix $\mathbf{G}_{i},i\in[1,M-1]$:
		\begin{equation}
			\label{eq:givensMatrix}
			\mathbf{G}_{i} =
			\begin{bmatrix}
				1 &\cdots &0&0&\cdots&0\\
				\vdots&\ddots&\vdots&\vdots&\ddots&\vdots\\
				0 &\cdots &c_{i}&-s_{i}&\cdots&0\\
				0 &\cdots &s_{i}&c_{i}&\cdots&0\\
				\vdots&\ddots&\vdots&\vdots&\ddots&\vdots\\
				0 &\cdots &0&0&\cdots&1\\
			\end{bmatrix},
		\end{equation}
		where 
		\begin{align}
			\label{eq:cs}
			c_{i}=\frac{\mathbf{H}(i)[i,i]}{\sqrt{\mathbf{H}(i)[i,i]^{2}+\mathbf{H}(i)[i+1,i]^{2}}},\\\notag	s_{i}=\frac{\mathbf{H}(i)[i+1,i]}{\sqrt{\mathbf{H}(i)[i,i]^{2}+\mathbf{H}(i)[i+1,i]^{2}}},\notag
		\end{align}
		and $\mathbf{H}(i)=\mathbf{G}^{T}_{i-1}\mathbf{H}(i-1), \mathbf{H}(1)=\mathbf{T}_{k-1}$. At the end of this iteration,
		\begin{equation}\notag
			\mathbf{T}_{k}=\mathbf{G}^{T}_{M-1}\cdots\mathbf{G}^{T}_{1}\mathbf{T}_{k-1}\mathbf{G}_{1}\cdots\mathbf{G}_{M-1}, 
		\end{equation}
		and $\mathbf{Q}_{k-1}=\mathbf{G}_{1}\cdots\mathbf{G}_{M-1}$.

		Next, we consider how to securely and efficiently perform the above process in secret sharing domain. 
		The main challenging part of the secure QR algorithm is to let {\csa} obliviously calculate square root and division in secret sharing. 
		This can be effectively addressed based on the techniques introduced in Section \ref{sec:dimRec}.
		In addition, we note that securely realizing the QR Algorithm consists mainly of matrix multiplications in the secret sharing domain.

		A straightforward method is to perform secret-shared value-wise multiplication, which would require $M^{3}$ multiplications and online communication of $2M^{3}$ ring elements.
		We note that a better choice is to work with vectorization \cite{mohassel2017secureml}, where the Beaver triples needed for secure multiplication are represented in a vectorized form, i.e., $(\mathbf{X},\mathbf{Y}, \mathbf{Z})$, where $\mathbf{Z}=\mathbf{X}\mathbf{Y}$; and $\mathbf{X}$ and $\mathbf{Y}$ play the role in masking the input matrices during secure multiplication.  
		Based on such vectorization trick, the online communication is reduced to only $2M^{2}$ ring elements. 
		\main~chooses to work with vectorization for secret-shared matrix multiplication.

		Algorithm \ref{algo:8} presents our basic design for securely realizing the QR algorithm.
		With $\llbracket\overline{ \mathbf{A}}\rrbracket^{A}$ as input, it outputs $\llbracket\mathbf{T}_{K}\rrbracket^{A}$ and $\llbracket\mathbf{S}\rrbracket^{A}$.
		Note that the top-$k$ diagonal elements of $\mathbf{T}_{K}$ represent the desired top-$k$ eigenvalues of $\mathbf{A}^{*}$.
		For the secret-shared eigenvectors $\llbracket\mathbf{S}\rrbracket^{A}$ of $\llbracket \mathbf{\overline{A}}\rrbracket^{A}$, they can be obliviously transformed to the corresponding eigenvectors of $\llbracket \mathbf{A}^{*}\rrbracket^{A}$ by $\llbracket\mathbf{V}\rrbracket^{A}=\llbracket\mathbf{P}\rrbracket^{A}\llbracket\mathbf{S}\rrbracket^{A}$, where $\llbracket\mathbf{P}\rrbracket^{A}$ is the output matrix of our secure Arnoldi method (or Lanczos method).


		\begin{algorithm}[!t]
			\caption{Secure QR Algorithm} 
			\label{algo:8}
			\begin{algorithmic}[1] 
				\REQUIRE The encrypted matrix $\llbracket\overline{ \mathbf{A}}\rrbracket^{A}$.
				\ENSURE The encrypted {\eigs} $\llbracket\mathbf{T}_{K}\rrbracket^{A}/\llbracket\mathbf{S}\rrbracket^{A}$ of $\llbracket\overline{ \mathbf{A}}\rrbracket^{A}$.
				\STATE $\llbracket \mathbf{T}_{0}\rrbracket^{A}=\llbracket\overline{ \mathbf{A}}\rrbracket^{A}$, $\llbracket\mathbf{S}\rrbracket^{A}=\llbracket\mathbf{I}\rrbracket^{A}$. \#$\mathbf{I}$ is an identity matrix.
				\FOR{$k\in[1,K]$}
				\STATE $\llbracket\mathbf{H}(1)\rrbracket^{A}=\llbracket\mathbf{T}_{k-1}\rrbracket^{A}$.
				\FOR{$i\in[1,M-1]$}
				\STATE $\llbracket x\rrbracket^{A}=(\llbracket \mathbf{H}(i)[i,i]\rrbracket^{A})^{2}+(\llbracket \mathbf{H}(i)[i+1,i]\rrbracket^{A})^{2}$.
				
				\# Calculate the inverse square root of Eq. \ref{eq:cs}:
				\FOR{$n\in[1,\Omega]$} 
				\STATE $\llbracket y_{n+1} \rrbracket^{A} =\frac{1}{2}\cdot \llbracket y_{n}\rrbracket^{A}\cdot(3-\llbracket x\rrbracket^{A}\cdot\llbracket y_{n}\rrbracket^{A}\cdot\llbracket y_{n}\rrbracket^{A})$.
				\ENDFOR 
				\STATE $\llbracket c_{i}\rrbracket^{A}=\llbracket \mathbf{H}(i)[i,i]\rrbracket^{A}\cdot \llbracket y_{\Omega+1} \rrbracket^{A}$.
				\STATE $\llbracket s_{i}\rrbracket^{A}=\llbracket \mathbf{H}(i)[i+1,i]\rrbracket^{A}\cdot \llbracket y_{\Omega+1} \rrbracket^{A}$.
				\STATE $\llbracket\mathbf{G}_{i}\rrbracket^{A}=\llbracket\mathbf{I}\rrbracket^{A}$. 
				\STATE $\llbracket\mathbf{G}_{i}[i,i]\rrbracket^{A}=\llbracket c_{i}\rrbracket^{A}$; $\llbracket \mathbf{G}_{i}[i+1,i+1]\rrbracket^{A}=\llbracket c_{i}\rrbracket^{A}$. \STATE$\llbracket \mathbf{G}_{i}[i+1,i]\rrbracket^{A}=\llbracket s_{i}\rrbracket^{A}$; $\llbracket \mathbf{G}_{i}[i,i+1]\rrbracket^{A}=\llbracket -s_{i}\rrbracket^{A}$. 
				\STATE $\llbracket\mathbf{H}(i+1)\rrbracket^{A}=\llbracket\mathbf{G}_{i}^{T}\rrbracket^{A}\llbracket\mathbf{H}(i)\rrbracket^{A}$.\label{alg:Tk1}

				\STATE $\llbracket\mathbf{S}\rrbracket^{A}=\llbracket\mathbf{S}\rrbracket^{A}\llbracket\mathbf{G}_{i}\rrbracket^{A}$. \label{alg:S}
				\ENDFOR
				\STATE $\llbracket\mathbf{T}_{k}\rrbracket^{A}=\llbracket\mathbf{H}(M)\rrbracket^{A}\llbracket\mathbf{G}_{1}\rrbracket^{A}\cdots\llbracket\mathbf{G}_{M-1}\rrbracket^{A}$. \label{alg:Tk2}
				\ENDFOR
			\end{algorithmic}
		\end{algorithm}


		\subsection{Optimizing the Secure QR Algorithm }
		\label{sec:vectorMaxtrix}
		
		We now show how to further optimize the basic design introduced above to achieve an efficiency boost.
		%
		Our key insight is to first reformulate the plaintext QR algorithm by simplifying the Givens rotation matrix (i.e., Eq. \ref{eq:givensMatrix}), and then identify the correlated multiplications and extract repetitive multiplicands to further save the cost. 
		We first simplify the Givens rotation matrix based on the observation: in each Givens rotation $\mathbf{G}^{T}_{i}\mathbf{H}(i)$ or $\mathbf{H}(i)\mathbf{G}_{i}$, only the $i$-th and $(i+1)$-th rows of $\mathbf{H}(i)$ are updated (recall Eq. \ref{eq:givensMatrix}). Therefore, to save computation, we can reduce the Givens rotation matrix $\mathbf{G}_{i}$ from Eq. \ref{eq:givensMatrix} to $\begin{bmatrix}
			c_{i}&-s_{i}\\
			s_{i}&c_{i}
		\end{bmatrix}$ (denoted as $\mathbf{g}_{i}$). Fig. \ref{fig:QRoptim} illustrates the optimized QR decomposition on a $4*4$ upper Hessenberg matrix utilizing a series of Givens rotations. Similarly, when calculating $\llbracket\mathbf{S}\rrbracket^{A}=\llbracket\mathbf{S}\rrbracket^{A}\llbracket\mathbf{G}_{i}\rrbracket^{A}$ (i.e., line \ref{alg:S} in Algorithm \ref{algo:8}), we can also reduce $\llbracket\mathbf{G}_{i}\rrbracket^{A}$ to $\llbracket\mathbf{g}_{i}\rrbracket^{A}$.

		After the above simplification, we have a new observation: $\llbracket\mathbf{g}_{i}\rrbracket^{A}$ will be used repeatedly in several matrix multiplications. In a recent study \cite{kelkar2022secure}, Kelkar \textit{et al.} point out that when one of the multiplicands in a number of (secret-shared) matrix multiplications stays constant, the constant multiplicand can be masked and then opened only once so as to achieve cost savings on communication.
		For example, supposed we need to multiply $\llbracket\mathbf{U}\rrbracket^{A}$ with $\llbracket\mathbf{V}_1\rrbracket^{A}$, $\llbracket\mathbf{V}_2\rrbracket^{A}$, $\cdots$, $\llbracket\mathbf{V}_k\rrbracket^{A}$ in the secret sharing domain.
		We only need a single matrix sharing $\llbracket\mathbf{X}\rrbracket^{A}$ for $\llbracket\mathbf{U}\rrbracket^{A}$, rather than $k$ matrix sharings as in directly using the Beaver's trick in vectorized form. 
		Therefore, {\main} can have {\csa} only mask $\llbracket\mathbf{g}_{i}\rrbracket^{A}$ once for all secret-shared multiplications. In addition, the masked $\llbracket\mathbf{g}^{T}_{i}\rrbracket^{A}$ can be directly achieved by letting {\csa} \textit{locally} transpose the masked $\llbracket\mathbf{g}_{i}\rrbracket^{A}$.

		\begin{figure}[!t]
		\centering
		\includegraphics[width=0.9\linewidth]{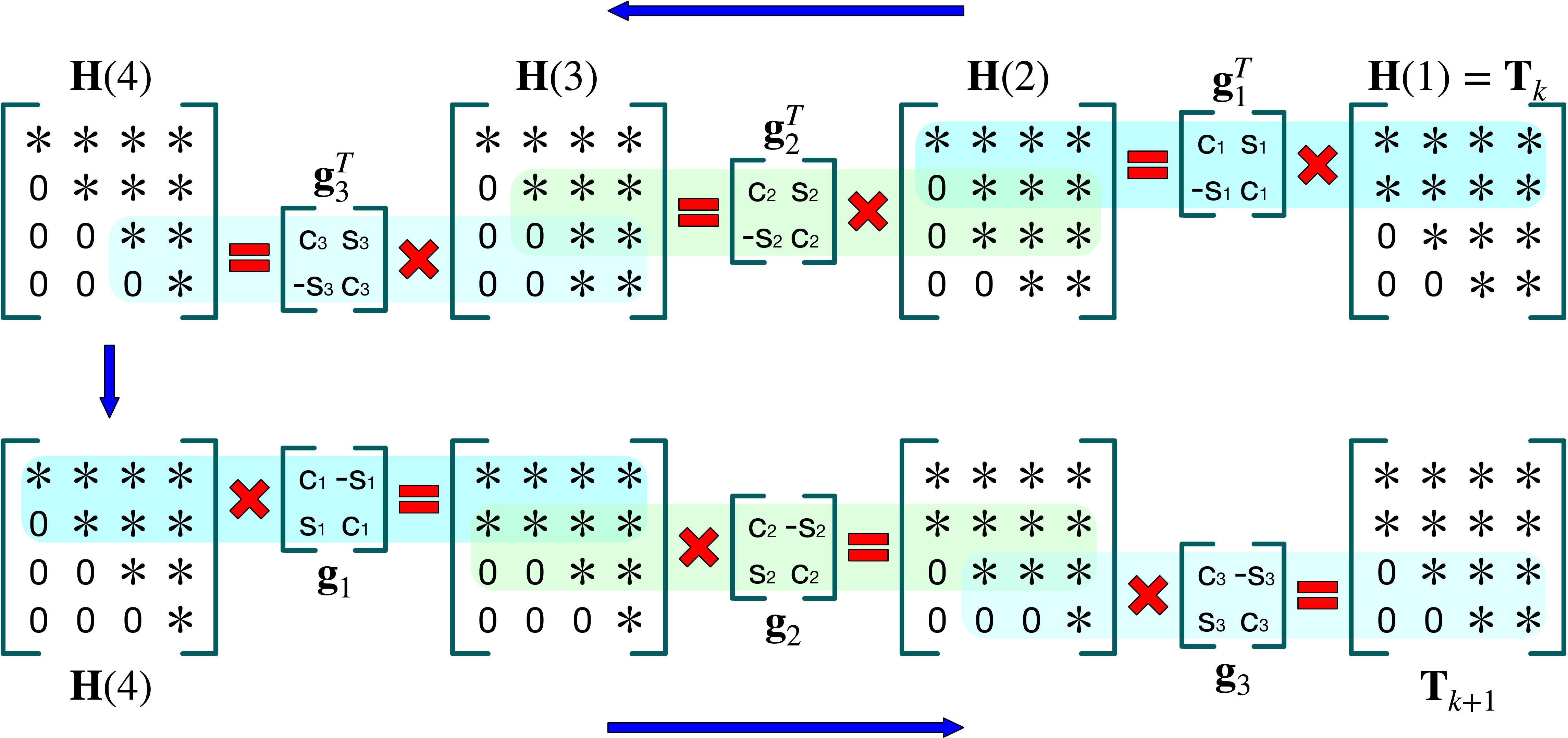}
		\caption{Illustration of the optimized QR decomposition on a $4*4$ upper Hessenberg matrix utilizing a series of Givens rotations.}
			\label{fig:QRoptim}
		\vspace{-10pt}
		\end{figure}


				
				
				%
				
		

		\noindent\textbf{Remark.} Our basic design for the secure QR algorithm requires {\csa} to online communicate $4K(M-1)M^{2}+ 2K(M-1)M^{2}=6K(M-1)M^{2}$ ring elements\footnote{We ignore the secure square root computation for simplicity.}. In contrast, our optimized secure QR algorithm only requires {\csa} to online communicate $K(M-1)(4M+4)+2K(M-1)M=K(M-1)(6M+4)$ ring elements. 
		As later shown in our experiments, the optimized secure QR algorithm can save up to $\mathbf{97\%}$ online communication as well as $\mathbf{9.3\%}$ computation cost as compared to the basic design.

		\subsection{Complexity Analysis}
		\label{sec:compexity_analy}
		
        The design of {\main} is comprised of several subroutines: 1) secure degree histogram estimation $\mathtt{secEst}$; 2) secure binning map generation $\mathtt{secBin}$; 3) local view data encryption $\mathtt{secEnc}$; 4) secure matrix dimension reduction $\mathtt{secRedu}$; 5) secure QR algorithm $\mathtt{secQR}$; 6) optimized secure QR algorithm $\mathtt{secOptiQR}$. Therefore, we analyze their theoretical performance complexities as follows.
		\begin{enumerate}[-]
			\item $\mathtt{secEst}$. Algorithm \ref{algo:3} describes $\mathtt{secEst}$ running between the cloud servers and some sampled users, which requires each of {\csa} to evaluate $S$ sampled users' DPF keys for all possible degrees $i\in[1,d_{max}]$, and thus its  computation complexity is $O(S\cdot d_{max})$. $\mathtt{secEst}$ does not require {\csa} to communicate with each other, but requires each of the $S$ sampled users to send a pair of DPF keys to {\csa}, and thus its  communication complexity is $O(S\cdot K)$, where $K$ is the size of DPF keys.
			\item  $\mathtt{secBin}$. Algorithm \ref{algo:4} describes $\mathtt{secBin}$ running at the cloud side, which requires {\csa} to securely evaluate for  all possible degrees $i\in[1,d_{max}]$, and thus its computation complexity is $O(d_{max})$. For communication, it is noted that we provide two secure comparison approaches (FSS-based and ASS-based). The FSS-based approach only requires {\csa} to send a masked value with bit length $l$ to each other, but the  ASS-based one requires  {\csa} to communicate $O(l\cdot\log l)$ bits.  Therefore, the communication complexity of FSS-based $\mathtt{secBin}$ is $O(l\cdot d_{max})$, and that of ASS-based $\mathtt{secBin}$ is $O(l\cdot\log l\cdot d_{max})$.
			\item  $\mathtt{secEnc}$. Algorithm \ref{algo:5} describes $\mathtt{secEnc}$ running at the users side, which requires each user $\mathcal{U}_{i}$ to encrypt his local view with length $|\mathcal{L}_{i}|$, and thus its computation complexity is $O(|\mathcal{L}_{i}|)$ and communication complexity is $O(|\mathcal{L}_{i}|\cdot l)$, where $l$ is the bit length of weight $\mathbf{A}[i,j]$. 
			\item  $\mathtt{secRedu}$. Algorithm \ref{algo:7} describes $\mathtt{secRedu}$ running at the cloud side, whose computation complexity is $O(N\cdot M)$, where $N$ is the width (or height) of the original matrix and $M$ is the width (or height) of the matrix after dimension reduction. Its communication complexity is $O(N\cdot M\cdot l)$. It is noted that the asymptotic computation and communication complexity of secure Lanczos method are also $O(N\cdot M)$ and $O(N\cdot M\cdot l)$ because its algorithm is similar to secure Arnoldi method (recall Algorithm \ref{algo:1} and \ref{algo:2}).
			\item  $\mathtt{secQR}$. Algorithm \ref{algo:8} describes $\mathtt{secQR}$ running at the cloud side, whose computation complexity is $O(K\cdot M^{3})$ and communication complexity is $O(K\cdot M^{3}\cdot l)$ where $K$ is the number of QR decomposition.		
			\item  $\mathtt{secOptiQR}$. It is noted that the asymptotic computation complexity of $\mathtt{secOptiQR}$ is  identical $\mathtt{secQR}$, i.e., $O(K\cdot M^{3})$, but its communication complexity is only $O(K\cdot M^{2}\cdot l)$ (recall the Remark in Section \ref{sec:vectorMaxtrix}).
		\end{enumerate}

		\section{Privacy and Security Analysis}
	\label{sec:security_analy}

	We now provide analysis regarding the protection \main~offers for the users. Note that {\main} aims to conceal the sensitive information regarding the non-zero elements in each user's local view vector, which includes the positions, values, and number (i.e., the degree).
	As \main~builds on LDP to perturb the exact degree of users and cryptographic techniques to safeguard the confidentiality of data values, we present our analysis in two parts.
	The first part is privacy analysis, which is to prove the differential privacy guarantee \main~offers for users.
	%
	%
	Note that blending in dummy edges in each user's local view vector as per the differential privacy mechanism obfuscates not only the number of non-zero elements but also their positions.
	The second part is security analysis, where we follow the standard simulation-based paradigm to prove data confidentiality against the cloud servers.
	Similar treatment also appears in prior works \cite{he2017composing,kacsmar2020differentially} using both DP and cryptography.

		\subsection{Privacy Analysis}
		\label{sec:privAnaly}
		\begin{theorem}
			\label{theo1}
			{\main} can achieve $(\epsilon, \delta)$-LDP for the degree of users in the same bin according to Definition \ref{def:LDP}.
		\end{theorem}
		\begin{proof}
		Given the degrees $d_{i},d_{j}$ of $\mathcal{U}_{i}$ and $\mathcal{U}_{j}$ who are in the same bin, and the sensitivity of the bin is $\Delta$. If both the noises drawn from $Lap(\epsilon,\delta,\Delta)$ by $\mathcal{U}_{i}$ and $\mathcal{U}_{j}$ are \textit{non-negative}, the probability of them outputting the same noisy degree $\hat{d}$ is bounded by
			\begin{align}\notag
				\frac{Pr[\hat{d}|d_{i}]}{Pr[\hat{d}|d_{j}]}=	\frac{Pr[\hat{d}-d_{i}]}{Pr[\hat{d}-d_{j}]}&=\frac{e^{\frac{\epsilon\cdot|\hat{d}-d_{i}-\mu|}{\Delta}}}{e^{\frac{\epsilon\cdot|\hat{d}-d_{j}-\mu|}{\Delta}}}\\\notag
				&=e^{\frac{\epsilon}{\Delta}\cdot(|\hat{d}-d_{i}-\mu|-|\hat{d}-d_{j}-\mu|)}\\\notag
				&\leq e^{\frac{\epsilon}{\Delta}\cdot(|\hat{d}-d_{i}|-|\hat{d}-d_{j}|)}\\\notag
				&\leq e^{\frac{\epsilon}{\Delta}|d_{i}-d_{j}|}\leq e^{\epsilon}.
			\end{align}
			In addition, we note that the probability to draw a \textit{negative} noise from $Lap(\epsilon,\delta,\Delta)$ is \cite{he2017composing}:
			\begin{equation}\notag
				Pr[x<0]=\sum_{x=-1}^{-\infty}\frac{e^{\frac{\epsilon}{\Delta}}-1}{e^{\frac{\epsilon}{\Delta}}+1}\cdot e^{\frac{\epsilon\cdot|x-\mu|}{\Delta}}=\frac{e^{\frac{-\mu\cdot\epsilon}{\Delta}}}{e^{\frac{\epsilon}{\Delta}}+1}.
			\end{equation}
			Given Eq. \ref{eq:mu} and $\Delta\ge0$, we have
			\begin{equation}\notag
				Pr[x<0]=1-(1-\delta)^{\frac{1}{\Delta}}\leq\delta,
			\end{equation}
			which means that the maximum probability of truncation is $\delta$. Therefore, with $1-\delta$, the probability to output the same noisy degree $\hat{d}$ from $\mathcal{U}_{i}$ and $\mathcal{U}_{j}$ is bounded by $e^{\epsilon}$, which satisfies $(\epsilon, \delta)$-LDP in Definition \ref{def:LDP}. 
			\end{proof}

		\subsection{Security Analysis}
		\label{subsec:security_analy}
		We follow the standard simulation-based paradigm to analyze the security of {\main}.  In this paradigm, a protocol is secure if the view of the corrupted party during the protocol execution can be generated by a simulator given only the party's input and legitimate output.
		Let $\Pi$ denote the protocol in \main~for secure \eig.
		Recall that the cloud servers \csa~neither provide input nor obtain output in \main.
		\begin{definition}
			\label{def:simu}
			Let $\mathrm{view}_{\mathcal{CS}_t}^{\Pi}$ denote each $\mathcal{CS}_t$'s view during the execution of $\Pi$. We say that $\Pi$ is secure in the semi-honest and non-colluding setting, if for each $\mathcal{CS}_{t}$ there exists a PPT simulator which can generate a simulated view such that $\mathrm{Sim}_{\mathcal{CS}_t}\mathop  \approx  \mathrm{view}_{\mathcal{CS}_t}^{\Pi}$.
			That is, the simulator can simulate a view for $\mathcal{CS}_t$, which is indistinguishable from $\mathcal{CS}_t$'s view during the execution of $\Pi$.
		\end{definition}
	
		\begin{theorem}
			\label{theo2}
			{\main} is secure according to Definition \ref{def:simu}.
		\end{theorem}
		
		\begin{proof}
		It is noted that {\main} invokes the subroutines $\mathtt{secEst}$, $\mathtt{secBin}$, $\mathtt{secEnc}$, $\mathtt{secRedu}$, and $\mathtt{secQR}$/$\mathtt{secOptiQR}$ in order.
			 If the simulator for each subroutine exists, then our complete protocol is secure \cite{canetti2000security,katz2005handling,curran2019procsa}. Since the roles of {\csa} in these subroutines are symmetric, it suffices to show the existence of simulators for $\mathcal{CS}_{1}$.
			
			\begin{enumerate}[-]
				
				\item  \textbf{Simulators for $\mathcal{CS}_{1}$ in $\mathtt{secEnc}$, $\mathtt{secRedu}$ and $ \mathtt{secQR}$.}  It is easy to see that the simulators for $\mathcal{CS}_{1}$ in these subroutines must exist, because they only require basic addition and multiplication over additive secret shares (which can be simulated by random values as per the security of additive secret sharing \cite{mohassel2017secureml}).	
				
				\item \textbf{Simulator for $\mathcal{CS}_{1}$ in $\mathtt{secEst}$.} $\mathcal{CS}_{1}$ only has the public number $S$ of sampled users at the beginning, and later receives DPF keys $\{k_{j,1}\}_{j\in[1,S]}$ from the sampled users. Since $\mathcal{CS}_{1}$ does not receive any other information apart from the DPF keys, the simulator for $\mathtt{secEst}$ can be trivially constructed by invoking the DPF simulator. Therefore, from the security of DPF \cite{boyle2016function}, the simulator for $\mathtt{secEst}$ exists.

				\item \textbf{Simulator for $\mathcal{CS}_{1}$ in $\mathtt{secBin}$.} In the $\mathtt{secBin}$ subroutine (Algorithm \ref{algo:4}), the steps that require {\csa} to interact are in lines \ref{alg:inter} and \ref{alg:reset}. Since each of them is invoked in order as per the processing pipeline and their inputs are secret shares, if the simulator for each of them exists, the simulator exists for $\mathcal{CS}_{1}$ in $\mathtt{secBin}$. 
				We first analyze the simulated view about line \ref{alg:inter}. As aforementioned, we provide two approaches (i.e., FSS-based and ASS-based) to perform the secure comparison operation in line \ref{alg:inter}. For the FSS-based approach, at the beginning of its each execution, $\mathcal{CS}_{1}$ has a DCF key $k_{1}$, a random value $r_{1}$,  a share $\langle accu\rangle_{1}$, and later receives a masked share $\langle accu\rangle_{2}+r_{2}$ followed by outputting $\mathsf{Eval}(k_{1},accu+r^{in})$. Since these information $\mathcal{CS}_{1}$ receives is all legitimate in FSS-based DCF, the simulator for $\mathtt{secBin}$ can be trivially constructed by invoking the simulator of DCF. From the security of FSS-based DCF \cite{boyle2021function}, the simulator for $\mathtt{secBin}$ exists. We then analyze the simulator for ASS-based approach. It is noted that the ASS-based approach is built on realization of a PPA in the secret sharing domain, which consists of basic secret-shared $\oplus$ and $\otimes$, and thus the simulator for it exists. Similarly, line \ref{alg:reset} also consists of basic multiplication, and thus the simulator for it exists \cite{mohassel2018aby3}.

				\item \textbf{Simulator for $\mathcal{CS}_{1}$ in $\mathtt{secOptiQR}$.} Our optimized secure QR algorithm builds on the technique from \cite{kelkar2022secure}. Specifically, when computing $k$ multiplications of the form $\llbracket\mathbf{U}\rrbracket^{A}\llbracket\mathbf{V}_j\rrbracket^{A},j\in[1,k]$, instead of using $k$ independent beaver triples $\llbracket\mathbf{Z}_j\rrbracket^{A}=\llbracket\mathbf{X}_j\rrbracket^{A}\llbracket\mathbf{Y}_j\rrbracket^{A}$, we can use $k$ correlated beaver triples $\llbracket\mathbf{Z}_j\rrbracket^{A}=\llbracket\mathbf{X}\rrbracket^{A}\llbracket\mathbf{Y}_j\rrbracket^{A}$, where a single matrix sharing $\llbracket\mathbf{X}\rrbracket^{A}$ is used to mask the constant multiplicand $\llbracket\mathbf{U}\rrbracket^{A}$. Next, we analyze the existence of the simulator for the technique.
				We first show that the view of {\cs} on the correlated Beaver triples can be simulated. It is noted that the shares of each element in the correlated beaver triples $\llbracket\mathbf{Z}_j\rrbracket^{A}=\llbracket\mathbf{X}\rrbracket^{A}\llbracket\mathbf{Y}_j\rrbracket^{A}$ is generated based on the standard secret sharing. This means that the distribution of the Beaver triple shares received by {\cs} is identically distributed with random values in the view of {\cs}. We then show that the view of {\cs} on the masked constant multiplicand can be simulated. In this phase, {\cs} first receives $\langle\mathbf{U}\rangle_{2}-\langle\mathbf{X}\rangle_{2}$ and $\langle\mathbf{V}_{j}\rangle_{2}-\langle\mathbf{Y}_{j}\rangle_{2}$ from {\css}, and then outputs $\langle\mathbf{Z}_{j}\rangle_{1},j\in[1,k]$. It is noted that $\langle\mathbf{V}_{j}\rangle_{2}-\langle\mathbf{Y}_{j}\rangle_{2}$ is a function of the correlated Beaver triple shares, $\mathbf{U}-\mathbf{X}$, and $\langle\mathbf{Z}_{j}\rangle_{1}$. Therefore, these matrices are uniformly random and independent from one another, and their joint distribution in both real view and simulated view is identical \cite{kelkar2022secure}. The simulator for $\mathtt{secOptiQR}$ exists.
				
			\end{enumerate}
		\end{proof}

		\section{Experiments}
		\label{sec:experiments}

		\subsection{Setup} 
		We implement a prototype system of {\main}~in Python. Our prototype implementation comprises $\sim$2000 lines of code (excluding the code of libraries). We also implement a test module with another 500 lines of code. All experiments are performed on a workstation with 16 Intel I7-10700K cores, 64GB RAM, 1TB SSD external storage running Ubuntu 20.04.2 LTS. The network bandwidth and latency are controlled by the local socket protocol.

		\begin{table}
			\normalsize 
			\centering
			\caption{Dataset Statistics}
			\label{Tab:dataset}
			\begin{tabular*}{\hsize}{@{}@{\extracolsep{\fill}}cccc@{}}
				\toprule
				Dataset &Type&Nodes&Edges\\\midrule
				Facebook&Undirected&3959&170,174\\ \hline
				Twitter&Directed&76244&1,768,149\\ \hline
				Google+& Directed&102100&13,673,453\\ 
				\bottomrule
			\end{tabular*}
			\vspace{-5pt}
		\end{table}
		
		\begin{figure}[!t]
			\centering
			\includegraphics[width=0.45\linewidth]{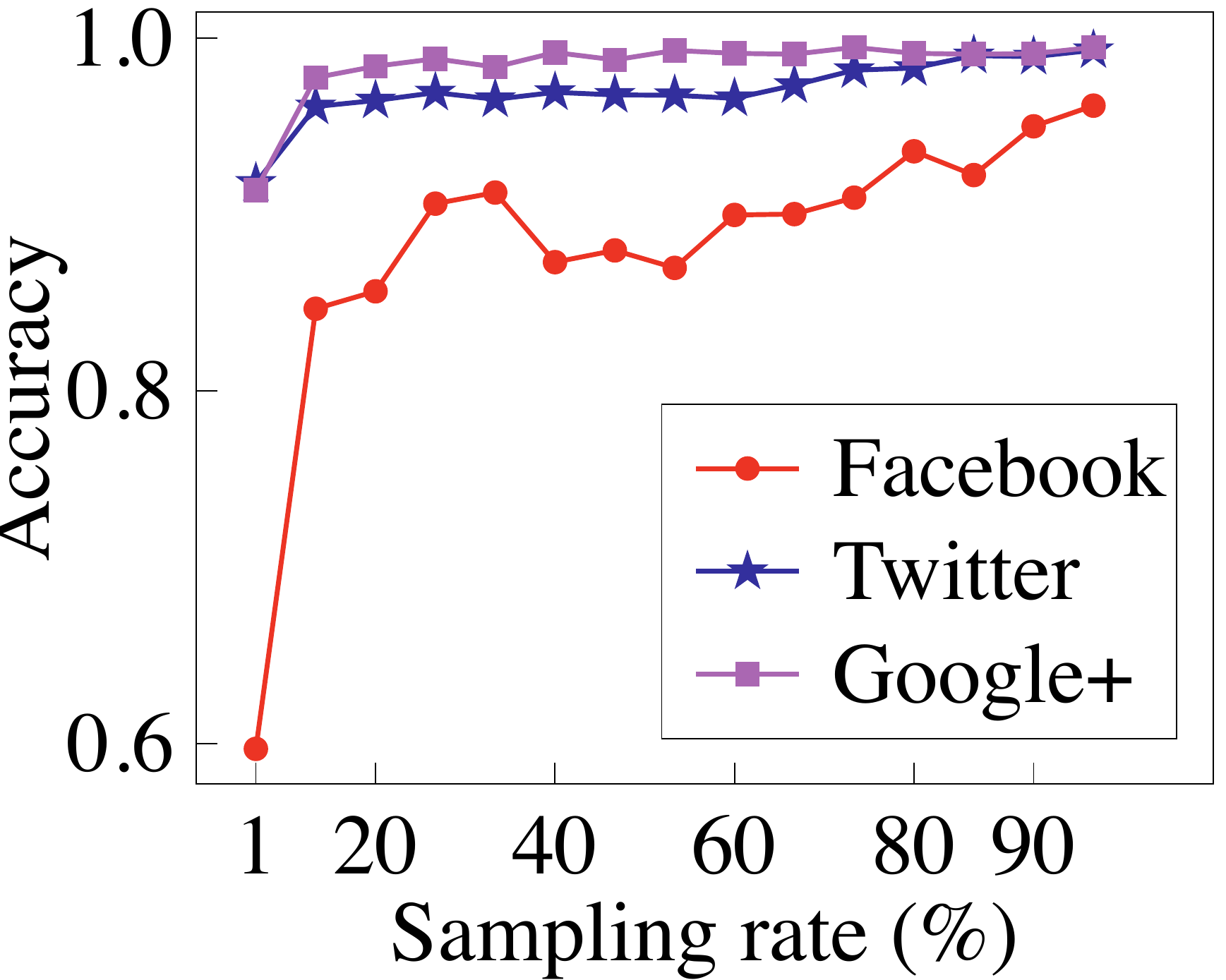}
			\caption{Accuracy of the binning map under different sampling rates.}
		\label{fig:groupAccuracy}
		\vspace{-12pt}
	\end{figure}

	\noindent\textbf{Social graph datasets.} We use three social graph datasets: Facebook\footnote{\url{http://snap.stanford.edu/data/ego-Facebook.html}.}, Twitter\footnote{\url{http://snap.stanford.edu/data/ego-Twitter.html}.}, and Google+\footnote{\url{http://snap.stanford.edu/data/ego-Gplus.html}.}. Their statistics are summarized in Table \ref{Tab:dataset}. 
	For each node in a tested graph dataset, a vector is extracted representing the node's local view on the whole graph, based on its social connections with other nodes.
	Recall that for privacy protection, {\main} aims to conceal the values, positions, and number of non-zero elements in each node's local view vector, as such information captures a node's private social interactions.
	In particular, in the context of these social graph datasets, the value information reflects two users are connected (e.g., they are friends), the position information reflects with which other users a user is connected, and the number information may reflect how many friends/followers a user has (indicating the user's social skills).

	\noindent\textbf{Protocol instantiation.} Eigendecomposition usually works on real numbers, while cryptographic computation needs to work with integers. Following priors works on secure computation \cite{mohassel2017secureml,mohassel2018aby3}, we use a common fixed-point encoding of real numbers. Specifically, for a private real number $x$, we consider a fixed-point encoding with $t$ bits of precision: $\lfloor x\cdot 2^{t} \rceil$.
	%
	In our experiments, we use the ring $\mathbb{Z}_{2^{32}}$ in the phase of secure degree histogram estimation and secure binning map generation, and the ring $\mathbb{Z}_{2^{64}}$ with $t=32$ bits of precision in the phase of secure {\eig}. The number of iterations of Eq. \ref{eq:square_root} is set to 25. 
	For DPF and DCF, we set the security parameter $\lambda$ to $128$. The size of the output matrix of Arnoldi method and Lanczos method is set to 15*15, and the top-3 {\eigs} are used to verify the accuracy of {\eig}, because only top-2 {\eigs} are used in most community detection tasks \cite{newman2006finding, wang2011identifying, newman2013spectral}.
	
	\noindent\textbf{Ground truth.} We use the standard Python library $\mathsf{scipy.sparse.coo\_matrix}$ to store the large-scale social graphs in sparse encoding (plaintext and ciphertext), and then use its standard {\eig} library $\mathsf{scipy.sparse.linalg.eigsh}$ and $\mathsf{scipy.sparse.linalg.eigs}$ to calculate the {\eigs} on symmetric matrix (i.e., Facebook) and non-symmetric matrices (i.e., Twitter and Google+), respectively. Subsequently, we will use the outputs of the standard library as the ground truth.
	
	\subsection{Evaluation on Accuracy}

	\noindent\textbf{Secure binning map generation.} To obtain an appropriate sampling rate when securely estimating the degree histogram, we compare the accuracy of the binning map with varying sampling rates. Fig. \ref{fig:groupAccuracy} summarizes the experiment results, where we set the number of bins to 10 and use the results of sampling rate $=100\%$ as the ground truth. It is observed that when the sampling rate is set to $10\%$, we can obtain a satisfactory accuracy (about $85\%-98\%$).
	
	\begin{figure}[t!]
		\begin{minipage}[t]{0.322\linewidth}
			\centering{\includegraphics[width=\linewidth]{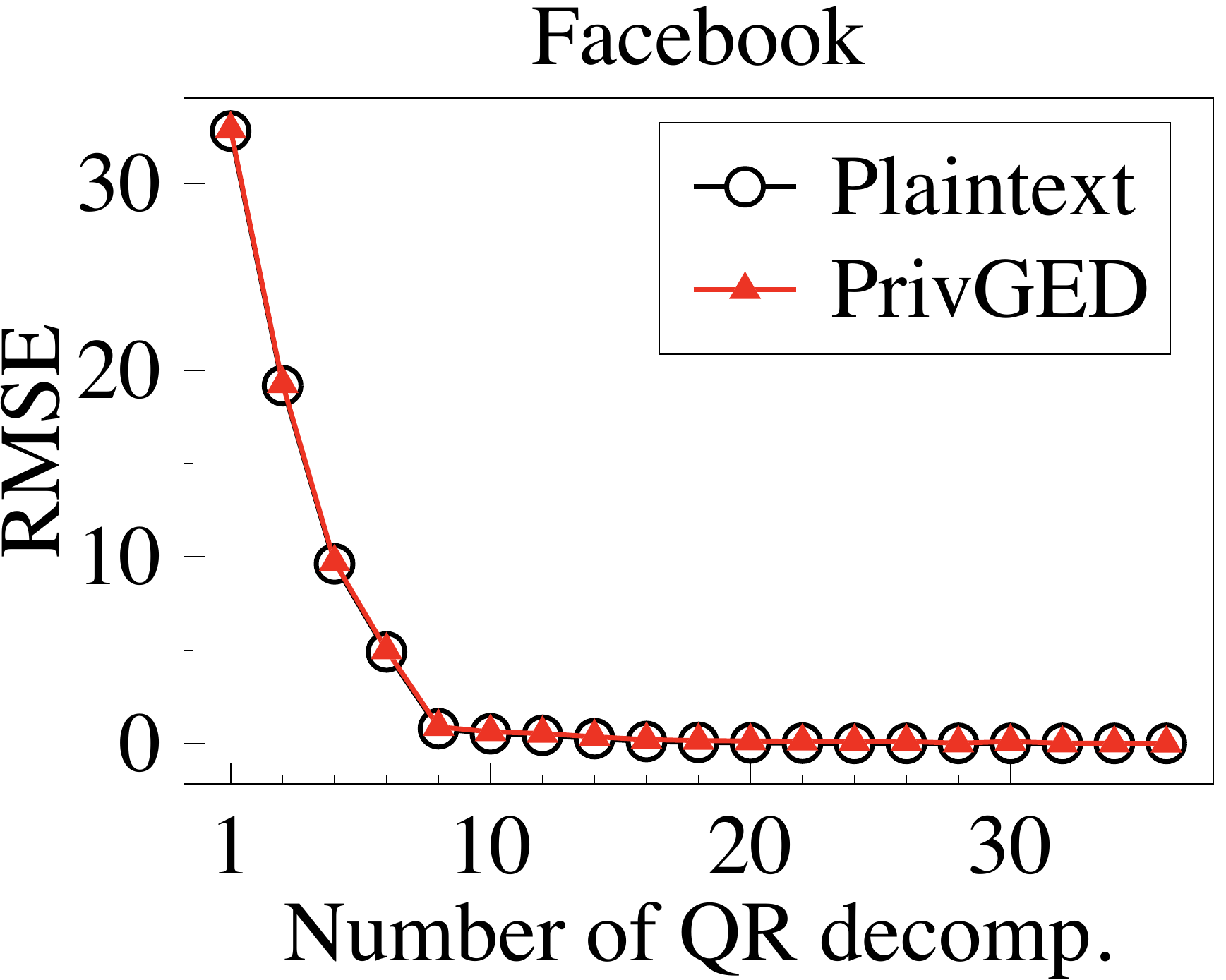}}
		\end{minipage}
		\begin{minipage}[t]{0.322\linewidth}
			\centering{\includegraphics[width=\linewidth]{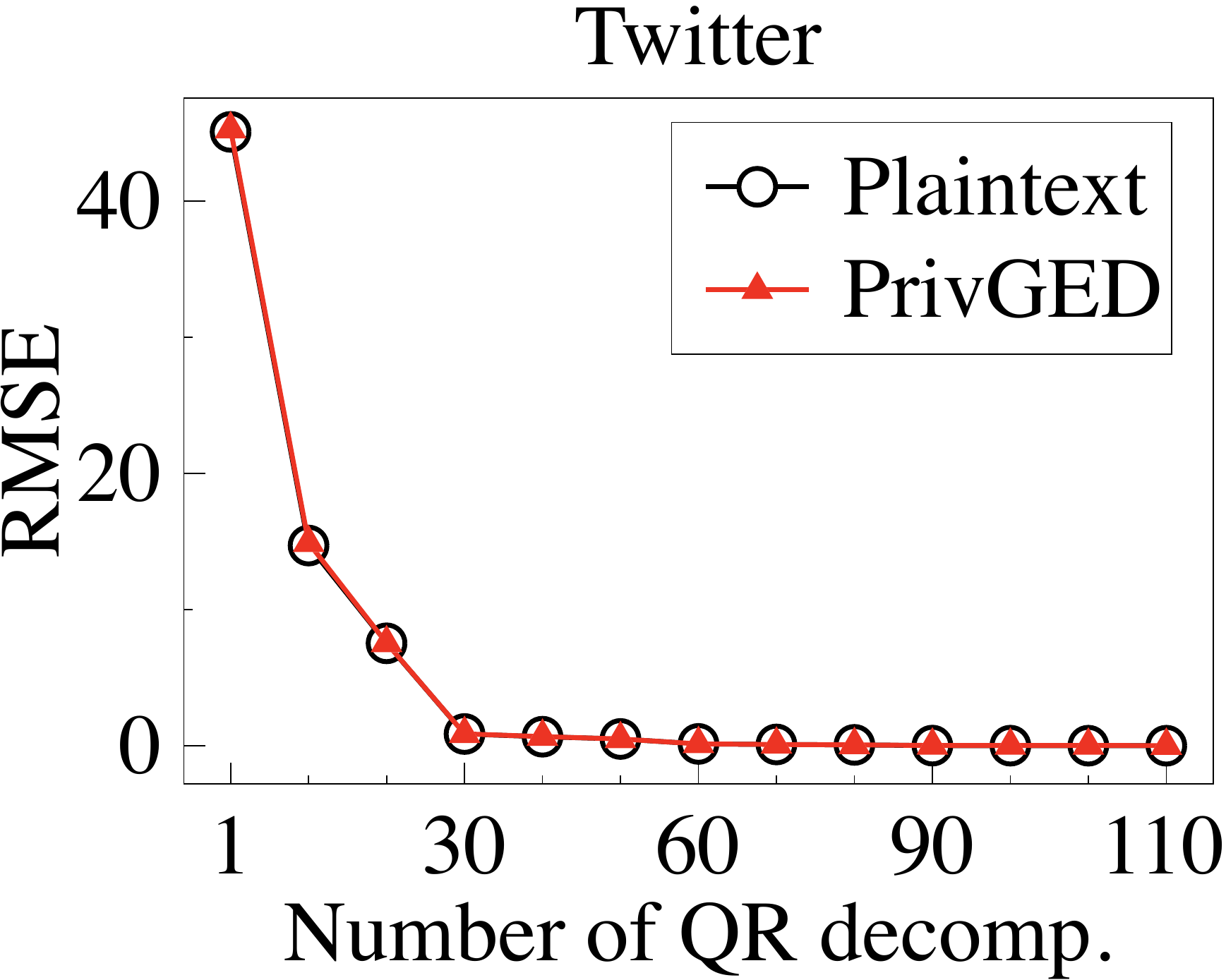}}
		\end{minipage}
		\begin{minipage}[t]{0.322\linewidth}
			\centering{\includegraphics[width=\linewidth]{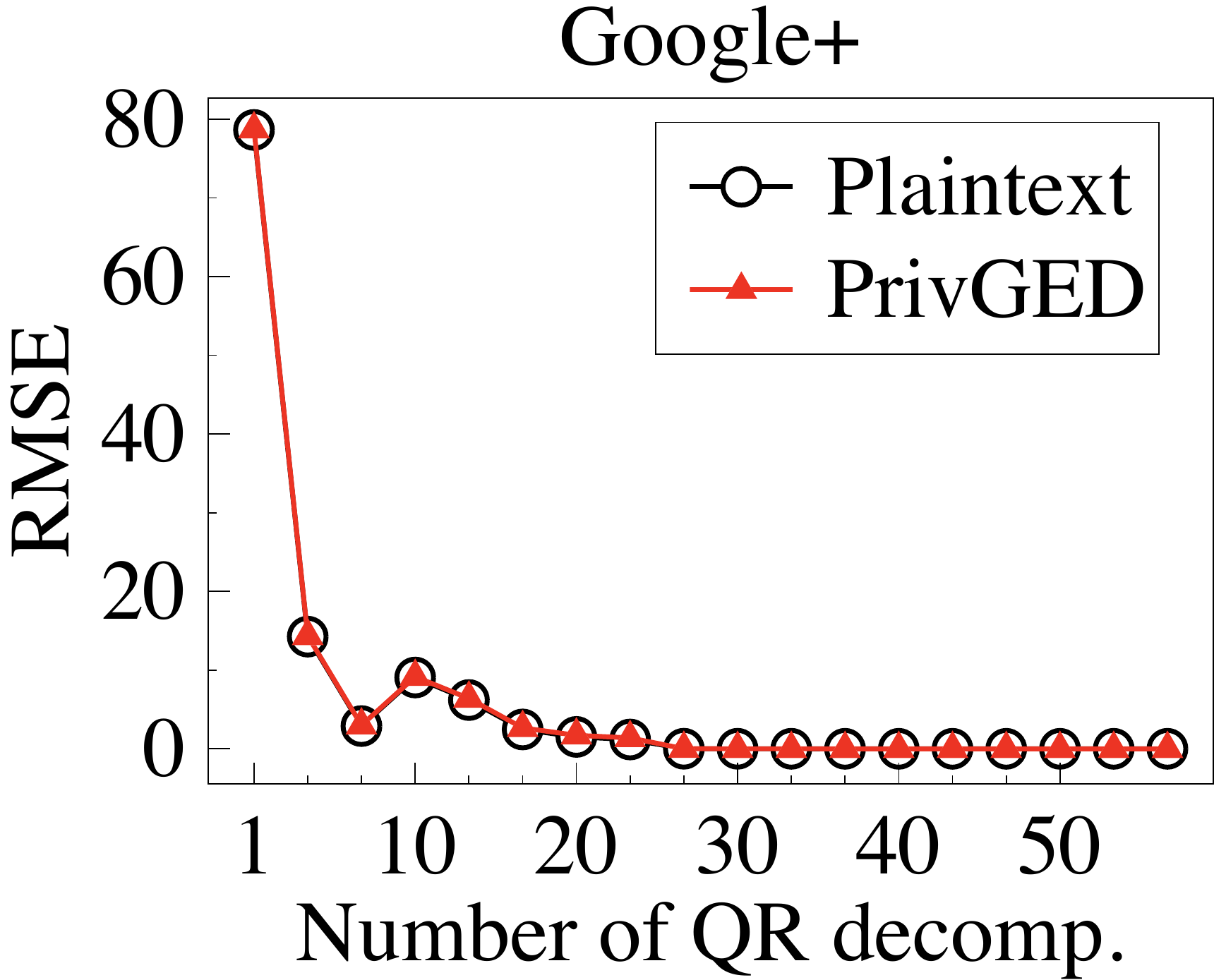}}
		\end{minipage}
		\caption{Results of RMSE between the top-3 eigenvalues in {\main} and plaintext, with varying number of QR decomposition.}
		\label{fig:RMSEValue}

	\end{figure}

	\begin{figure}[t!]
		\begin{minipage}[t]{0.322\linewidth}
			\centering{\includegraphics[width=\linewidth]{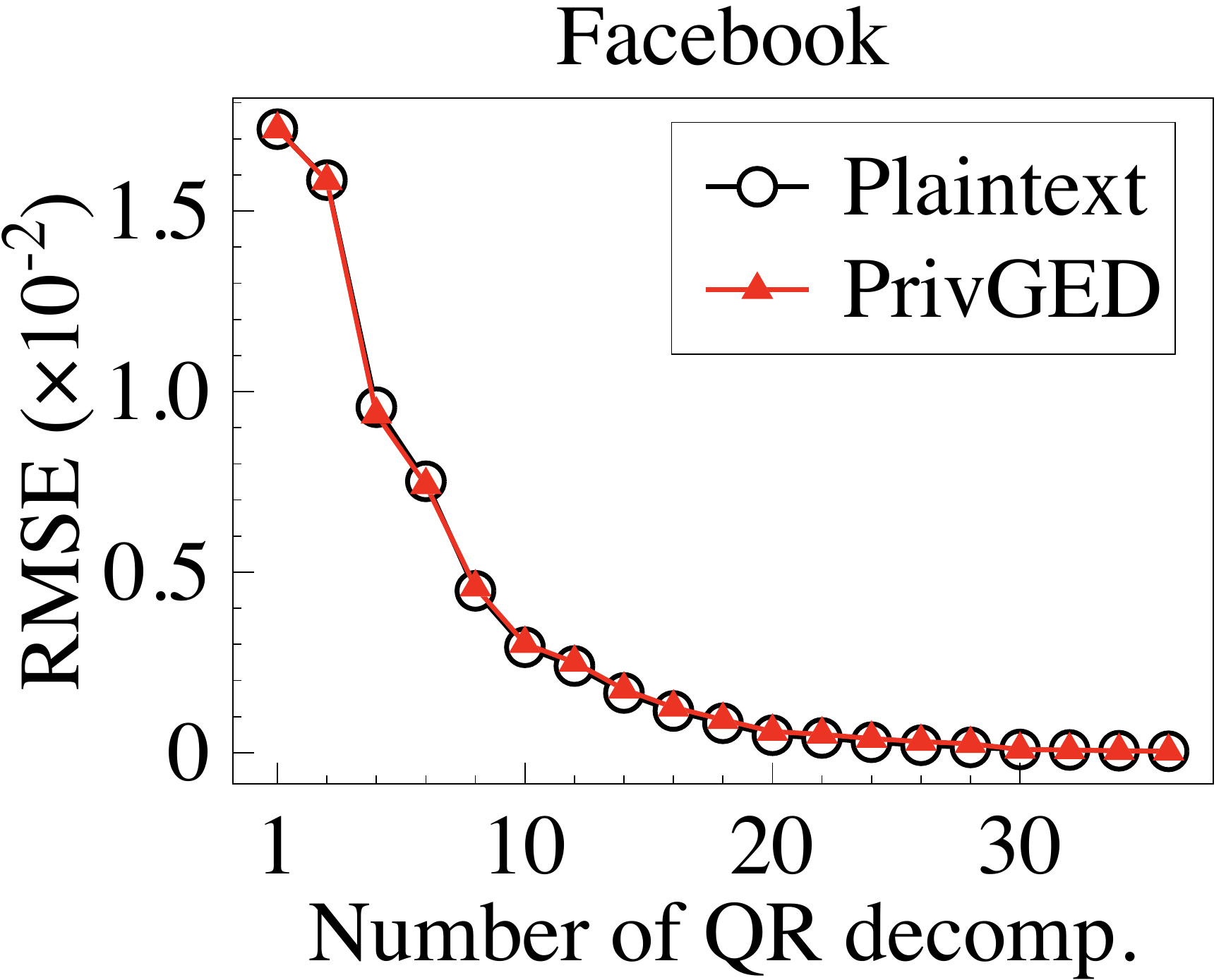}}
		\end{minipage}
		\begin{minipage}[t]{0.322\linewidth}
			\centering{\includegraphics[width=\linewidth]{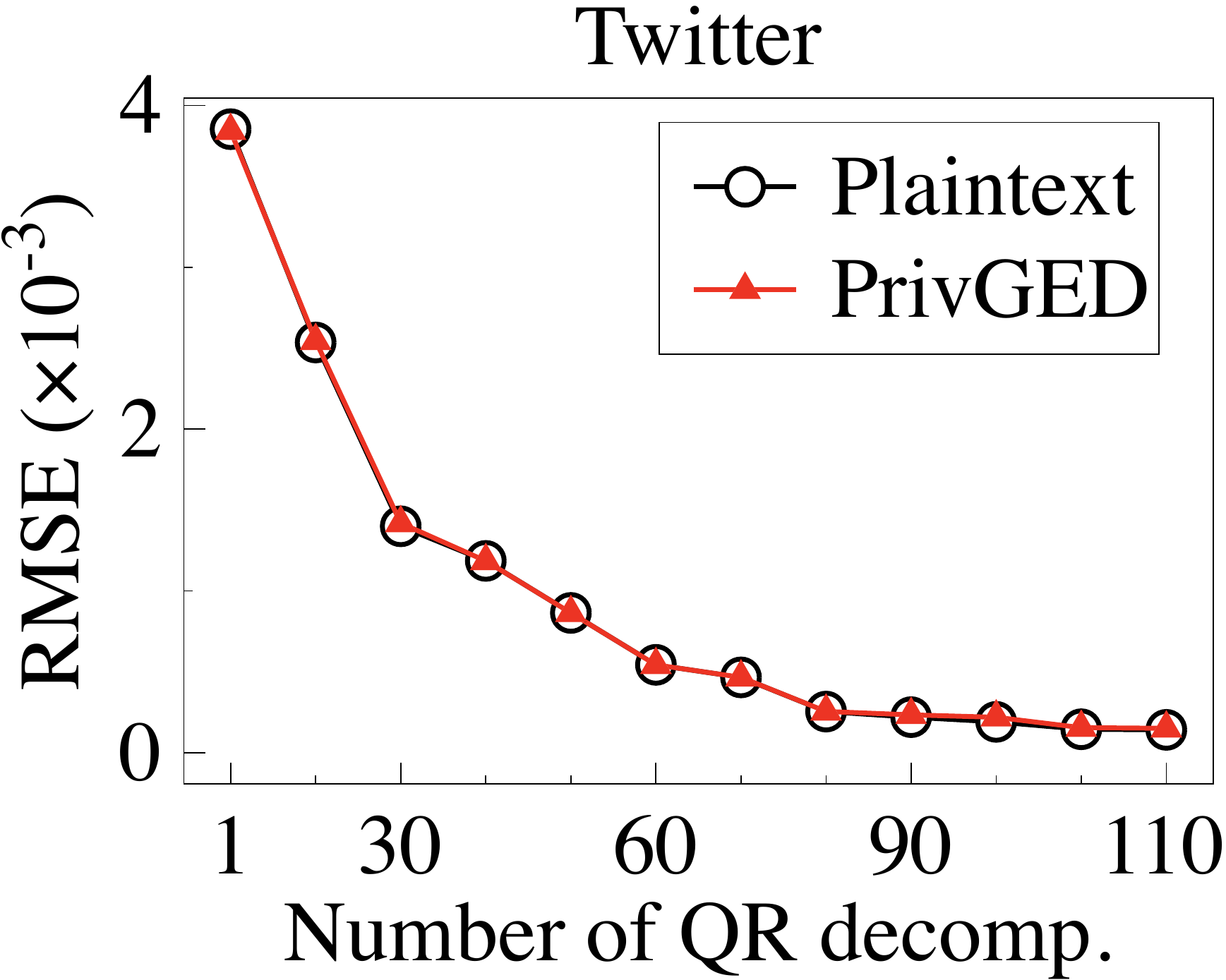}}
		\end{minipage}
		\begin{minipage}[t]{0.322\linewidth}
			\centering{\includegraphics[width=\linewidth]{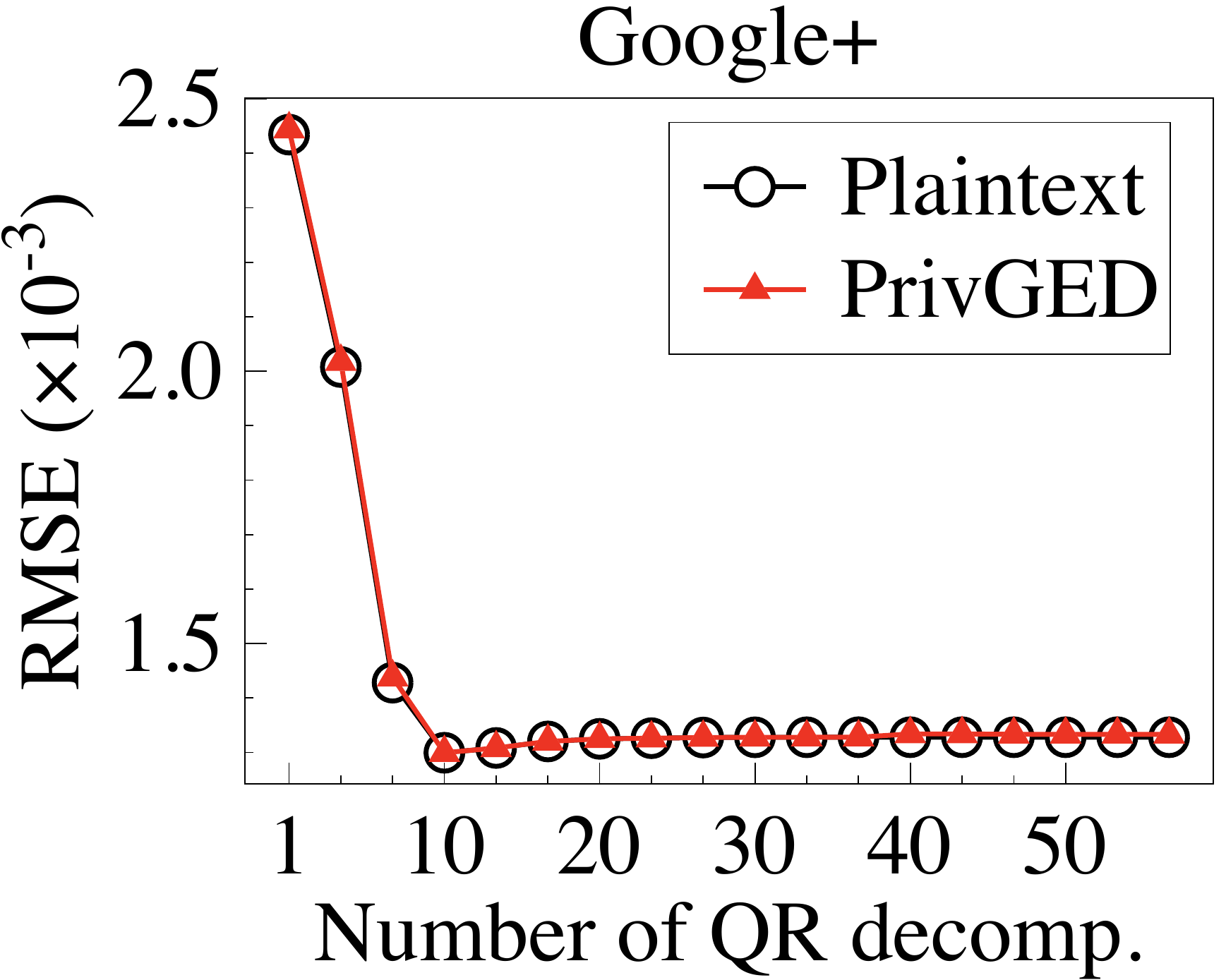}}
		\end{minipage}
		\caption{Results of RMSE between the top-3 eigenvectors in {\main} and plaintext, with varying number of QR decomposition.}
		\label{fig:RMSEVector}

	\end{figure}
	
	\noindent\textbf{Secure {\eig}.} We first implement the plaintext Arnoldi method, Lanczos method, and QR algorithm to compute the plaintext {\eigs}, and then calculate the RMSE in the top-3 {\eigs} between the plaintext and the ground truth (output by the standard Python library). We then execute {\main} to compute the ciphertext {\eigs}, and calculate the RMSE in the top-3 {\eigs} between the \main's result and the ground truth. We summarize the results in Fig. \ref{fig:RMSEValue} and Fig. \ref{fig:RMSEVector}.
	It is observed that the RMSE of {\main} is slightly higher than that of plaintext (about $0.1\%-5\%$), but they exhibit consistent behavior.

	\subsection{Evaluation on Performance}
	\label{exp:cost}
	\begin{figure}[!t]
		\centering
		\begin{minipage}[t]{0.4\linewidth}
			\centering{\includegraphics[width=\linewidth]{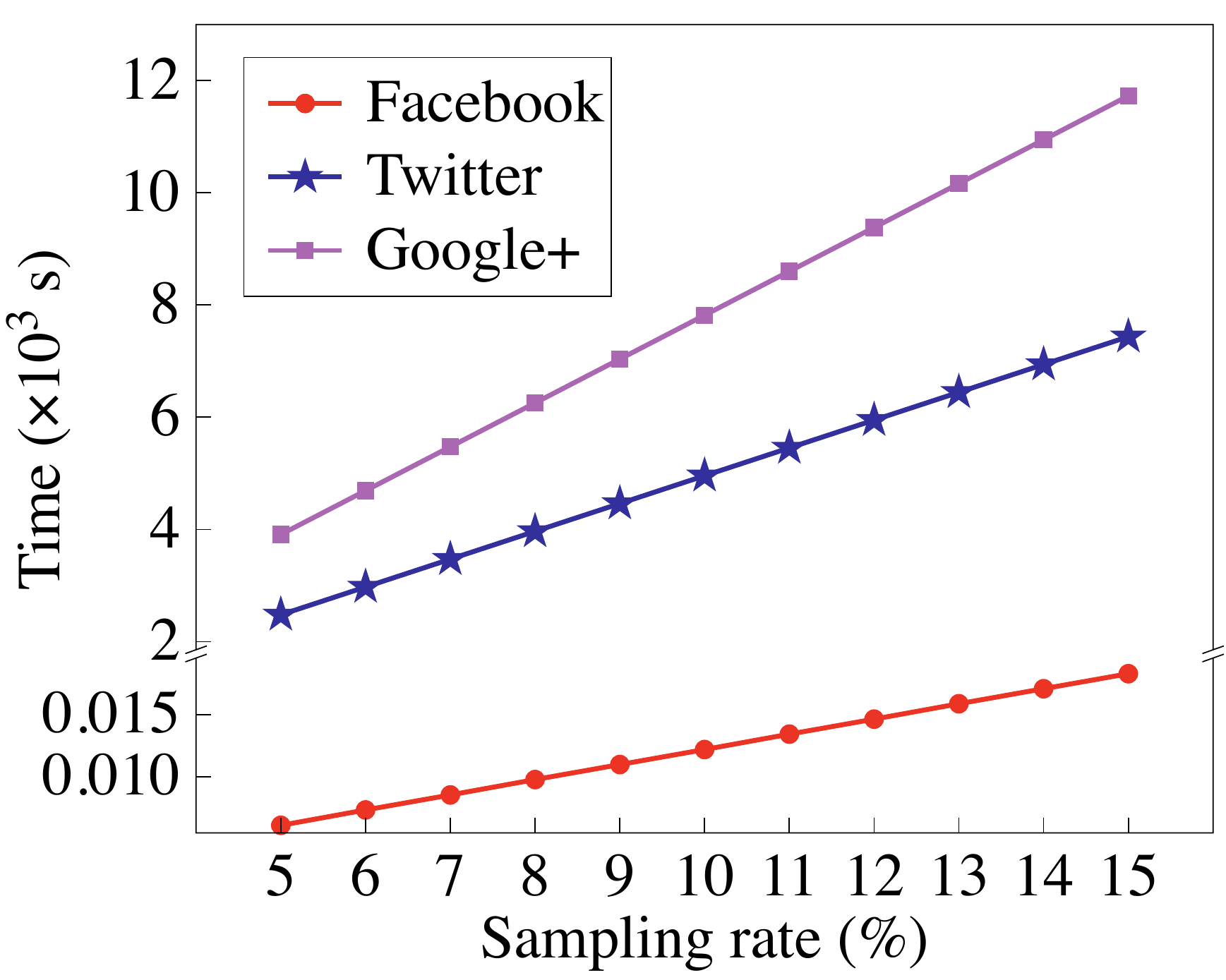}}
			\caption{Running time of secure degree histogram estimation at cloud.}
		\label{fig:histogramTime}
	\end{minipage}
	\begin{minipage}[t]{0.4\linewidth}
		\centering{\includegraphics[width=\linewidth]{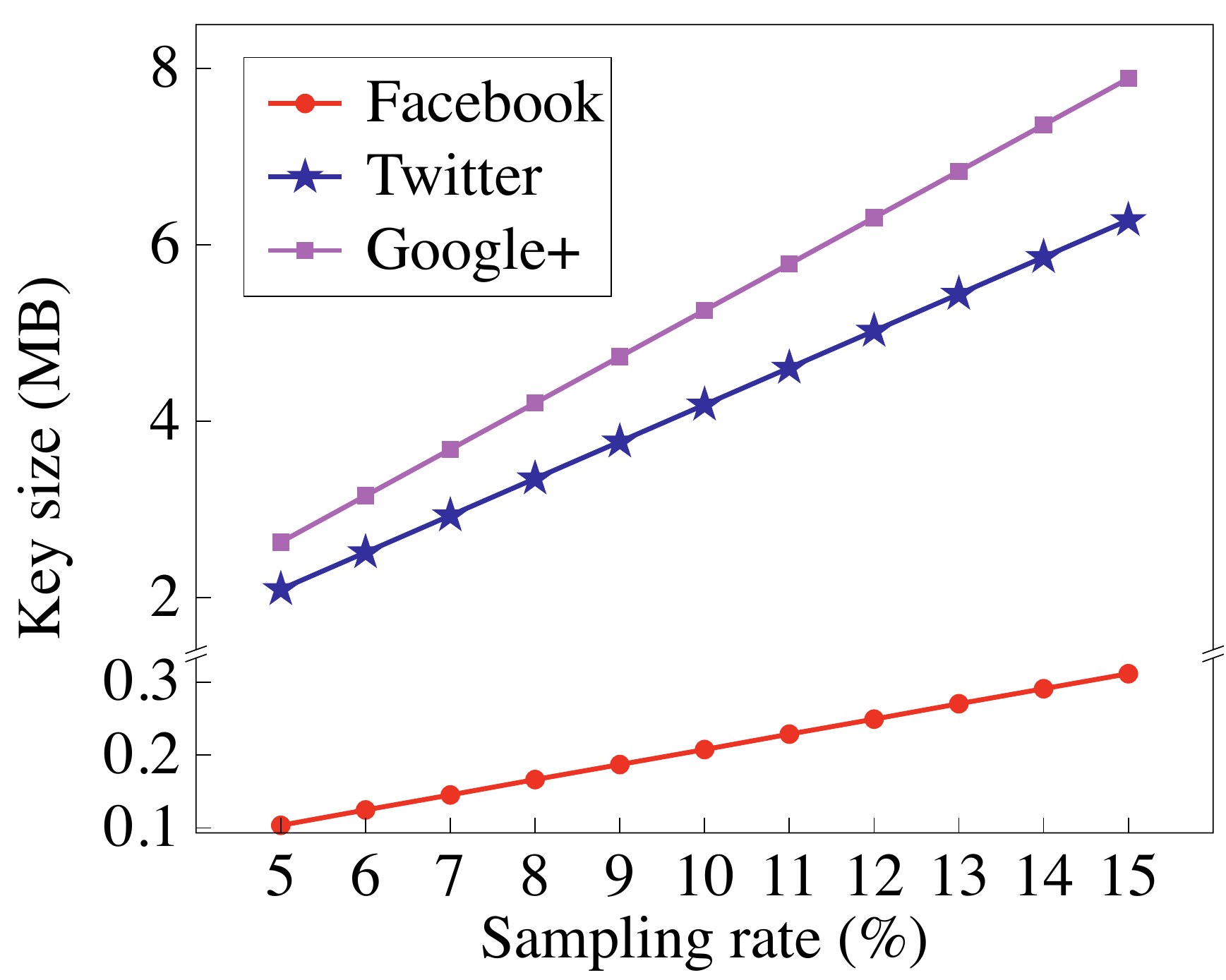}}
		\caption{DPF key size in secure degree histogram estimation.}
	\label{fig:histogramKeySize}
\end{minipage}
\vspace{-10pt}
\end{figure}

\noindent\textbf{Secure degree histogram estimation.} We first report {\main}’s cost in securely estimating the degree histogram. Fig. \ref{fig:histogramTime} illustrates the running time and Fig. \ref{fig:histogramKeySize} shows the key size of DPF involved in the process, where we set the number of bins to 10 and the maximum degree to $N/20$. It is noted that Algorithm \ref{algo:3} does not require communication between {\csa}. Over the three social graphs, the total key size of all sampled users only ranges from 0.1 MB to 7.8 MB.


\begin{figure}[!t]
\begin{minipage}[t]{0.32\linewidth}
	\centering{\includegraphics[width=\linewidth]{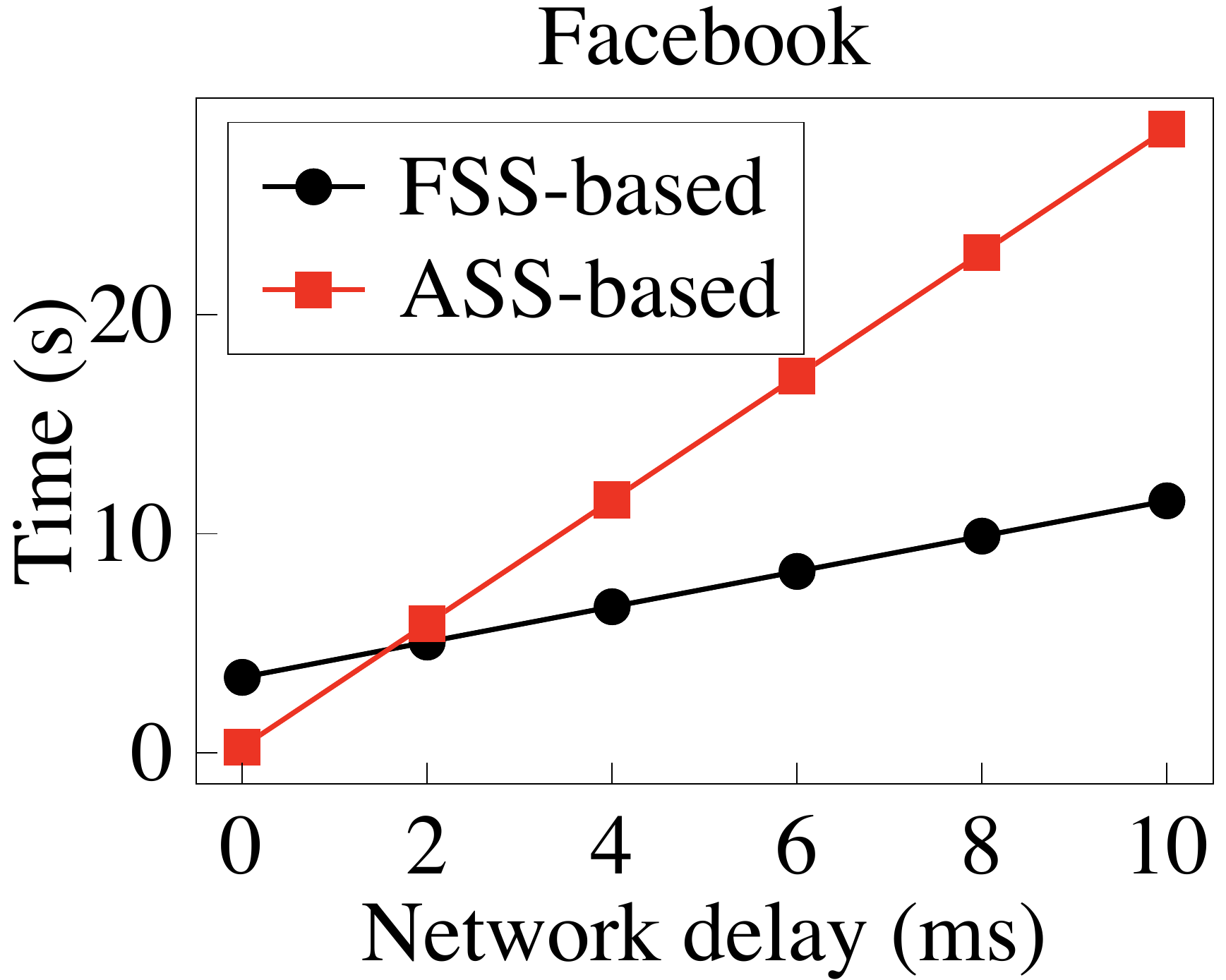}}
\end{minipage}
\begin{minipage}[t]{0.32\linewidth}
	\centering{\includegraphics[width=\linewidth]{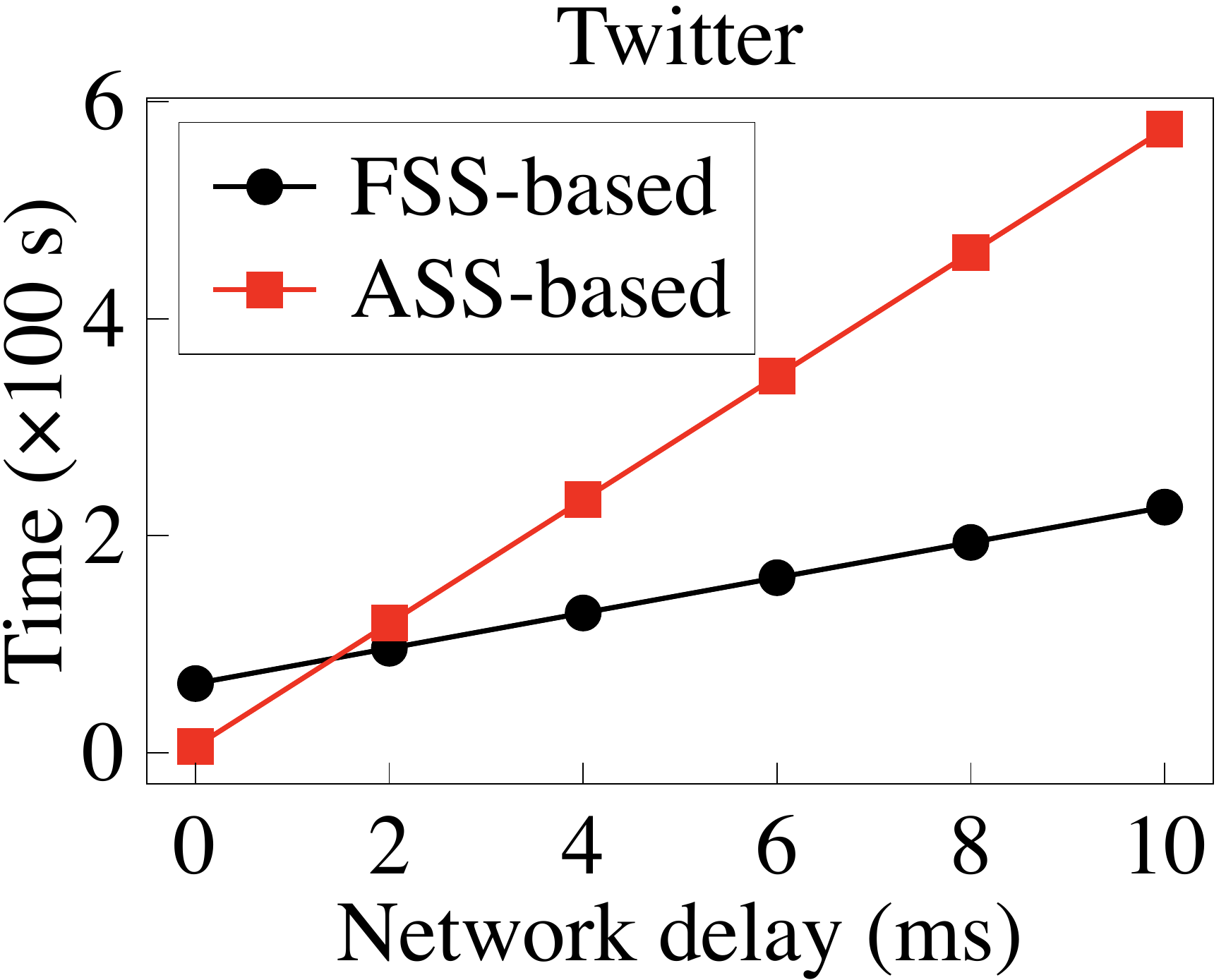}}
\end{minipage}
\begin{minipage}[t]{0.32\linewidth}
	\centering{\includegraphics[width=\linewidth]{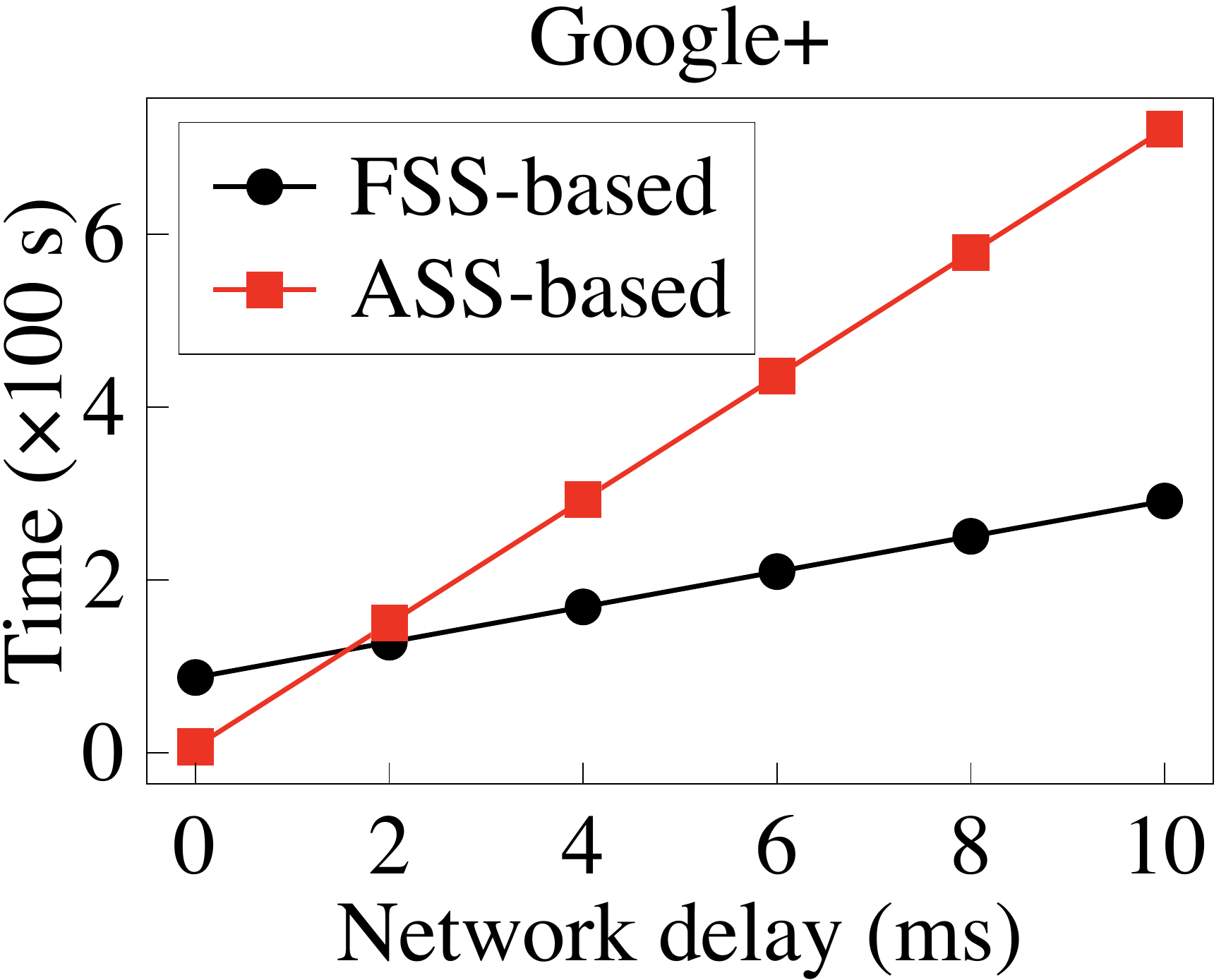}}
\end{minipage}
\caption{Running time of secure binning map generation with varying network delays, under different secure comparison approaches.}
\label{fig:binningTime}
\vspace{-10pt}
\end{figure}

\begin{figure}[!t]
\centering
\begin{minipage}[t]{0.45\linewidth}
\centering{\includegraphics[width=\linewidth]{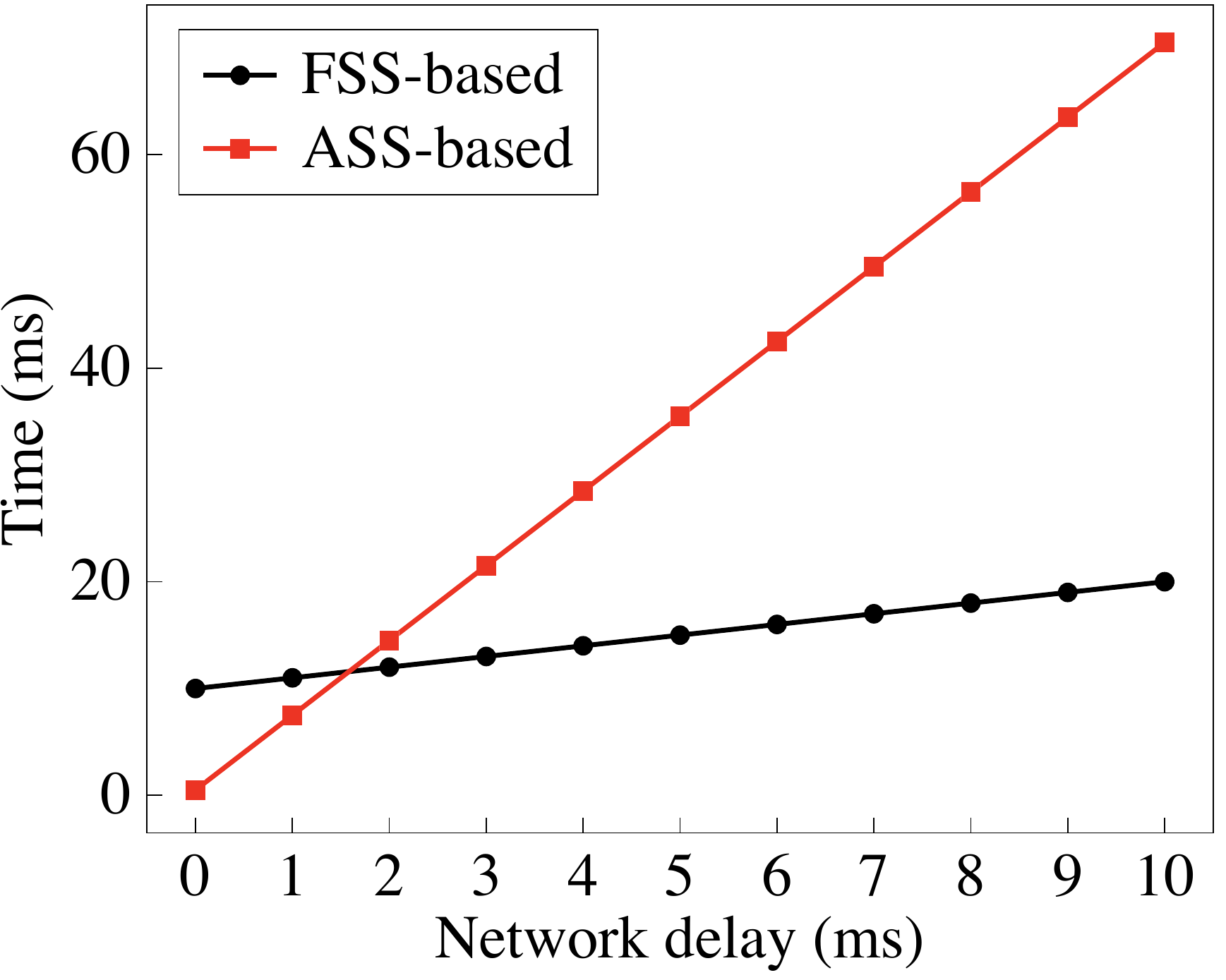}}
\caption{Running time comparison of ASS- and FSS-based secure comparison approaches.}
\label{fig:compTime}
\end{minipage}
\begin{minipage}[t]{0.45\linewidth}
\centering{\includegraphics[width=\linewidth]{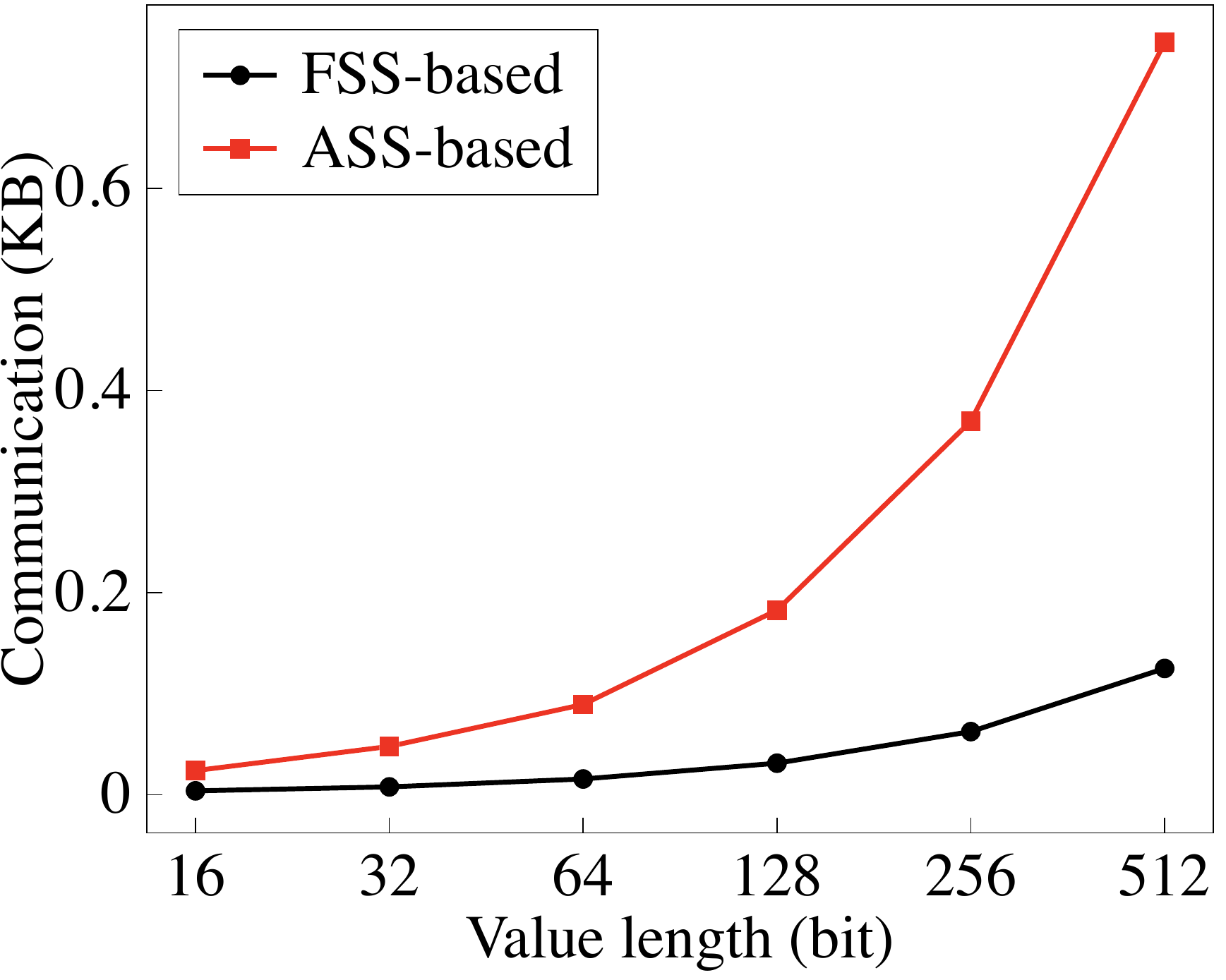}}
\caption{Bandwidth comparison of ASS- and FSS-based secure comparison approaches.}
\label{fig:compComm}
\end{minipage}
\end{figure}

\noindent\textbf{Secure binning map generation.} We then report {\main}’s performance in securely generating the binning map. Recall that in Section \ref{sec:1-2}, we provide an FSS-based approach and ASS-based approach to support the secure comparison operation under different network environments. Therefore, we first evaluate the efficiency under different network delays, and report the results in Fig. \ref{fig:binningTime}. Additionally, we report the benchmark cost of the two secure comparison approaches under different network delays and bit lengths, in Fig. \ref{fig:compTime} and Fig. \ref{fig:compComm}. It is observed that when the network delay is low (about 0 ms-2 ms), the ASS-based approach faster than the FSS-based approach.
As the network delays grows (i.e., $>$2 ms), the FSS-based approach becomes faster.

\begin{figure}[!t]
\centering
\begin{minipage}[t]{0.45\linewidth}
\centering{\includegraphics[width=\linewidth]{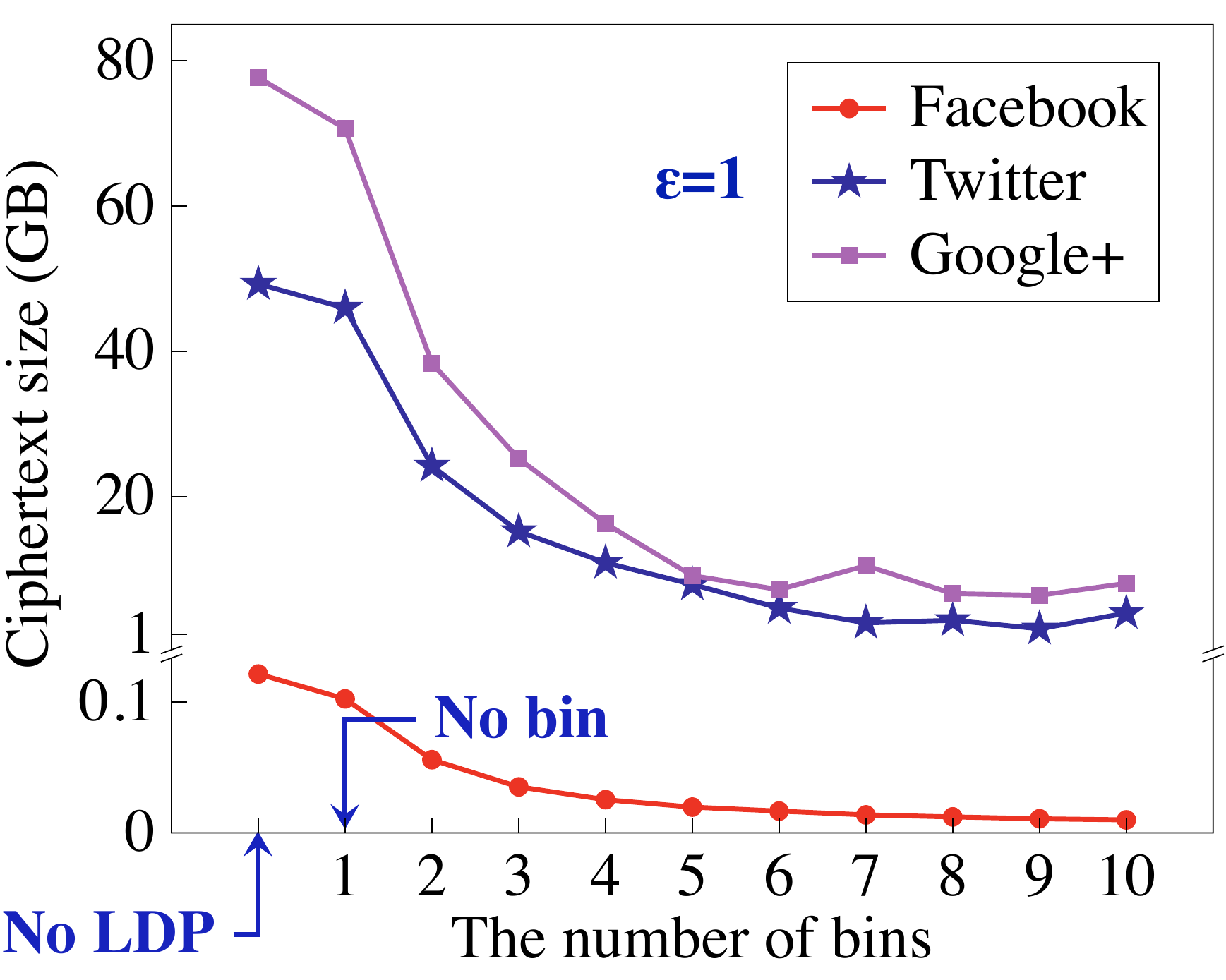}}
\end{minipage}
\begin{minipage}[t]{0.45\linewidth}
\centering{\includegraphics[width=\linewidth]{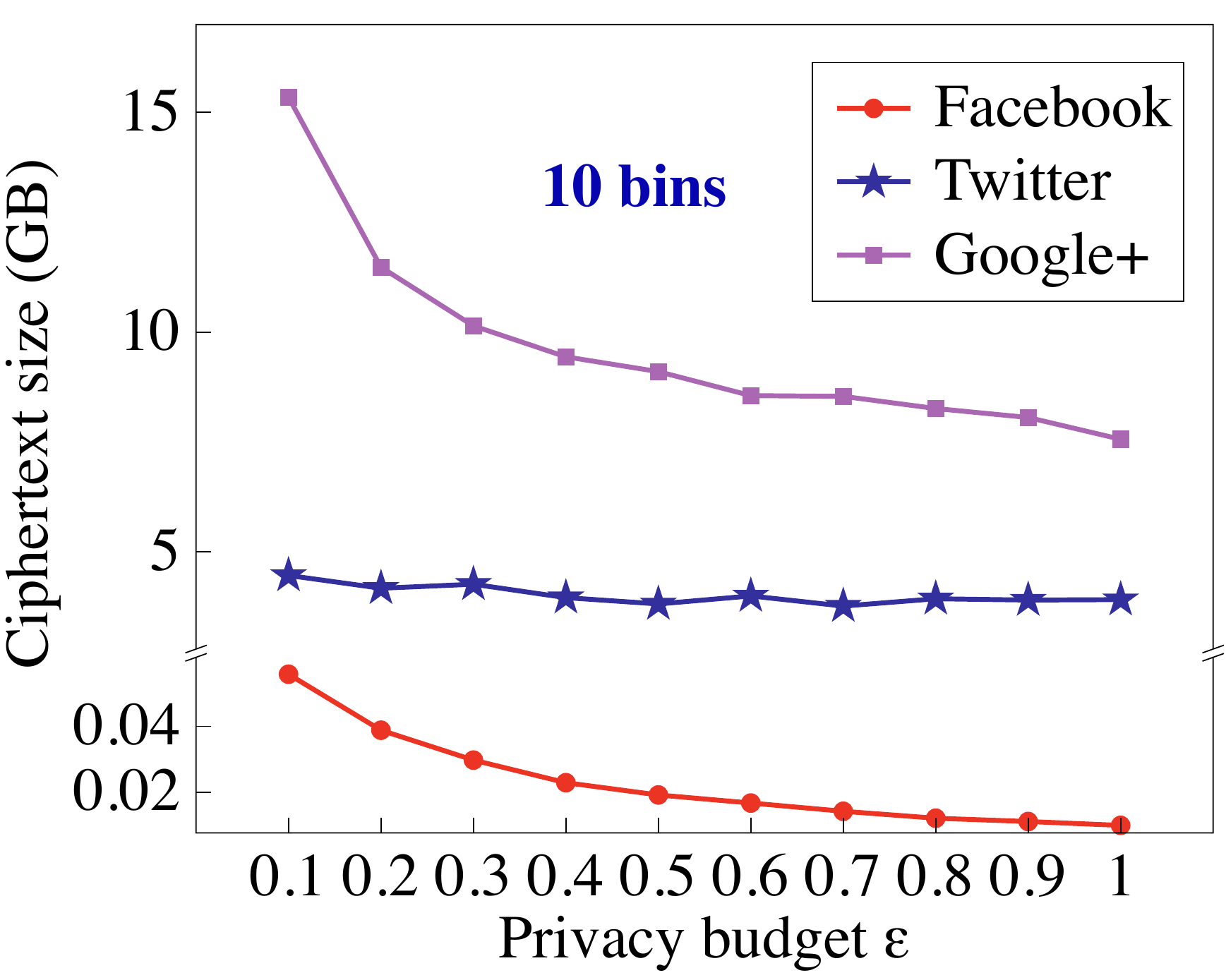}}
\end{minipage}
\caption{Total size of collected encrypted local views, for varying number of bins (left) and varying privacy budget $\epsilon$ (right) over different datasets (with $\delta=10^{-6}$). ``No bin" refers to the use of standard LDP.}
\label{fig:ciphertextSize}
\vspace{-10pt}
\end{figure}

\noindent\textbf{Local view data encryption.} We now report the size of the encrypted social graph to show the storage saving of {\main}, which is shown in Fig. \ref{fig:ciphertextSize}. 
It is observed that compared to direct encryption of all matrix elements by ASS, our protocol achieves a considerable storage saving (up to $\mathbf{90\%}$ under $\epsilon=1,\delta=10^{-6}$, and 10 bins). In addition, compared to way of standard LDP-based encryption (i.e., 1 bin), our encryption protocol (with 10 bins) still achieves a considerable storage saving (up to $\mathbf{70\%}$).

\begin{figure}[!t]
\centering
\begin{minipage}[t]{0.45\linewidth}
\centering{\includegraphics[width=\linewidth]{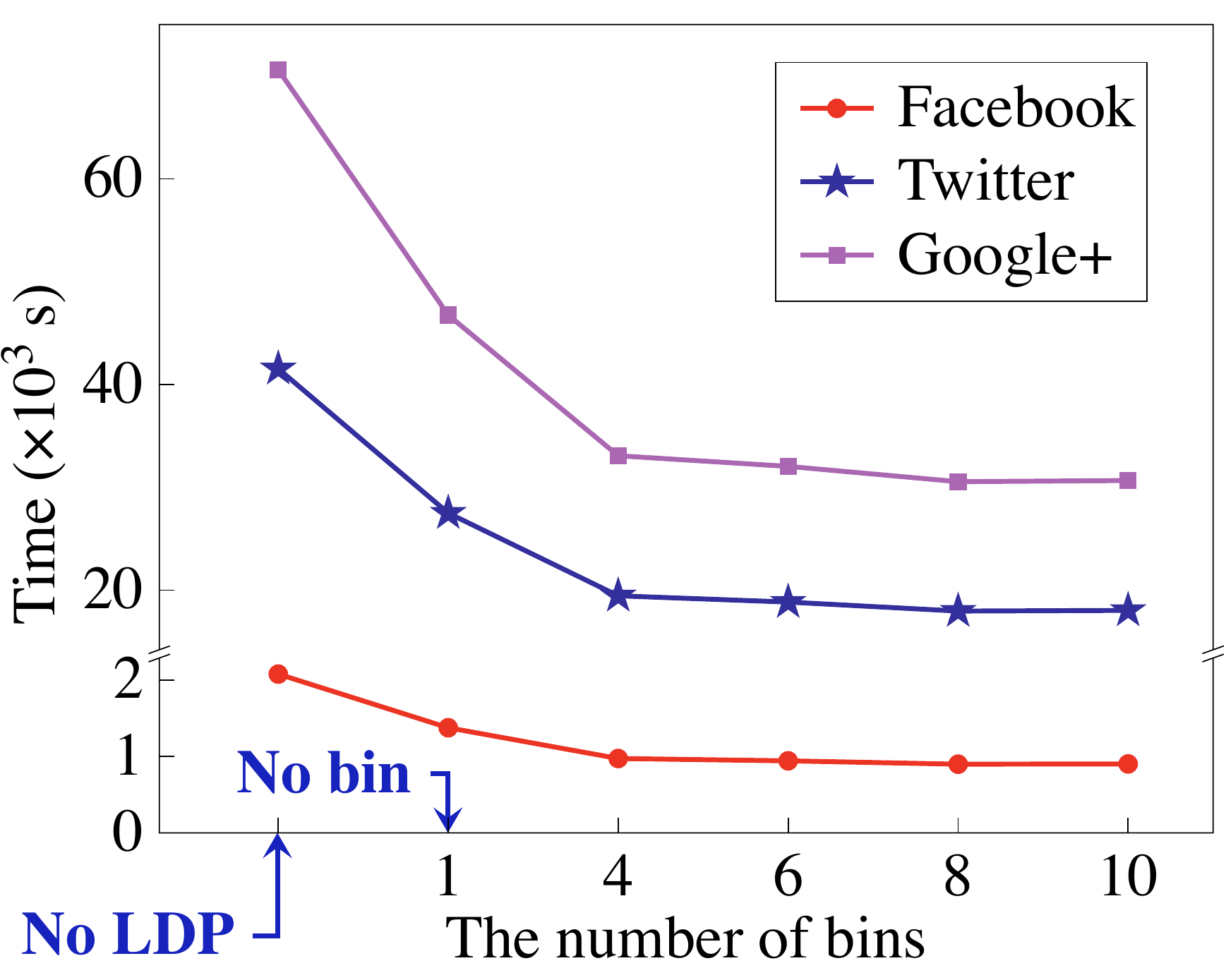}}
\end{minipage}
\begin{minipage}[t]{0.45\linewidth}
\centering{\includegraphics[width=\linewidth]{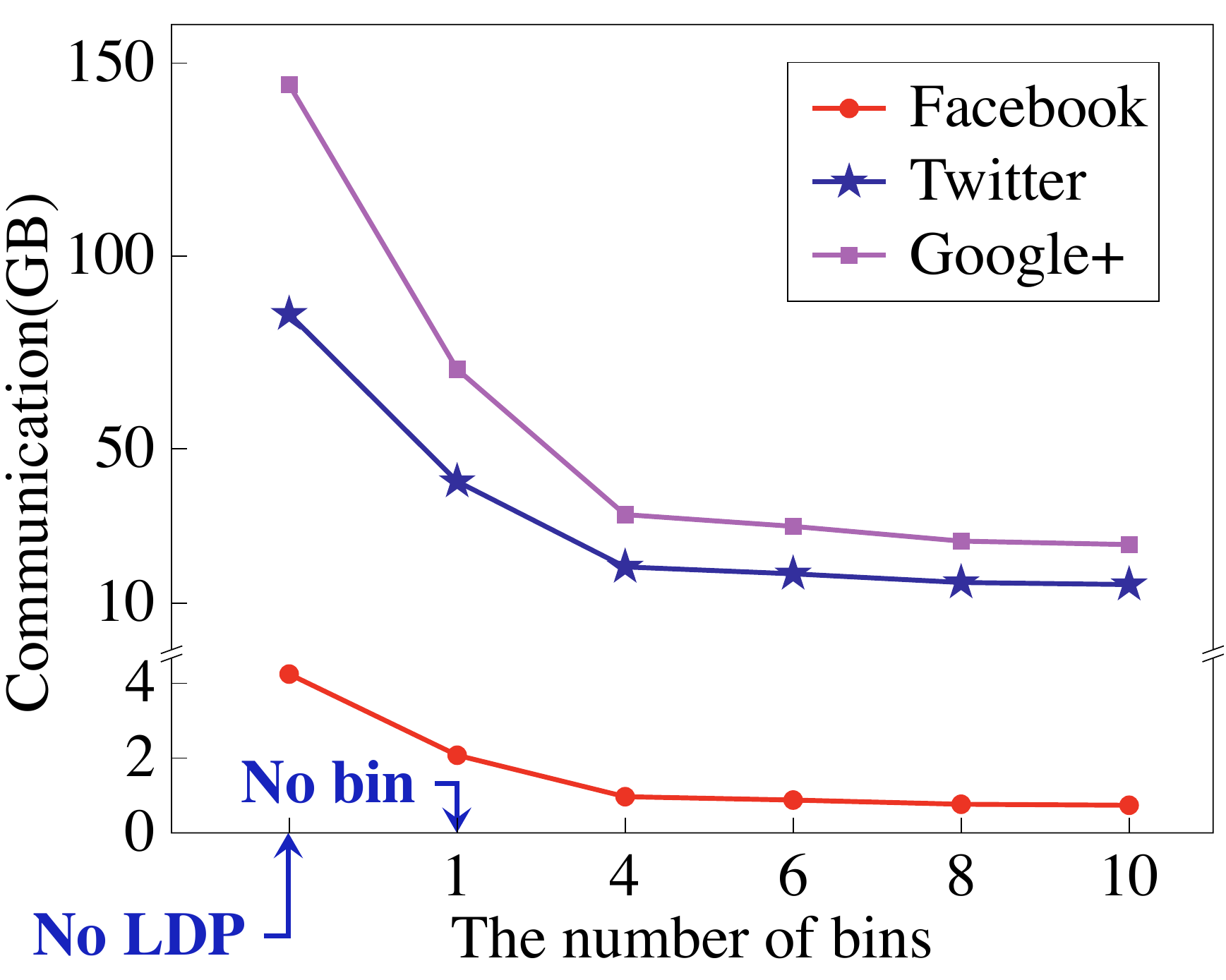}}
\end{minipage}
\caption{Running time and online communication of secure {\eig} at the cloud, for varying number of bins ($\epsilon=1,\delta=10^{-6}$). ``No bin" refers to the use of standard LDP.}
\label{fig:perforDiffBin}
\vspace{-10pt}
\end{figure}

\begin{figure}[!t]
\centering
\begin{minipage}[t]{0.45\linewidth}
\centering{\includegraphics[width=\linewidth]{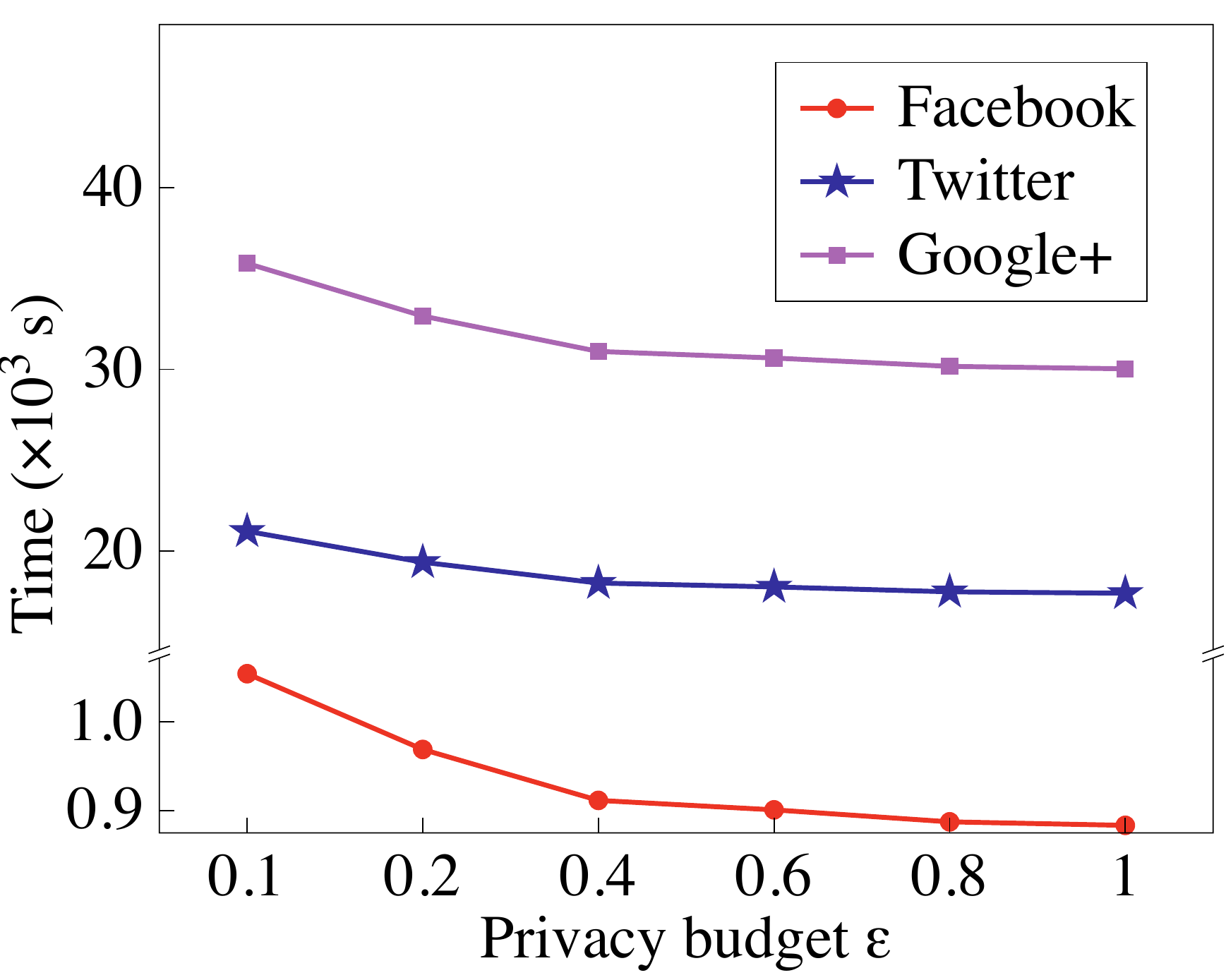}}
\end{minipage}
\begin{minipage}[t]{0.45\linewidth}
\centering{\includegraphics[width=\linewidth]{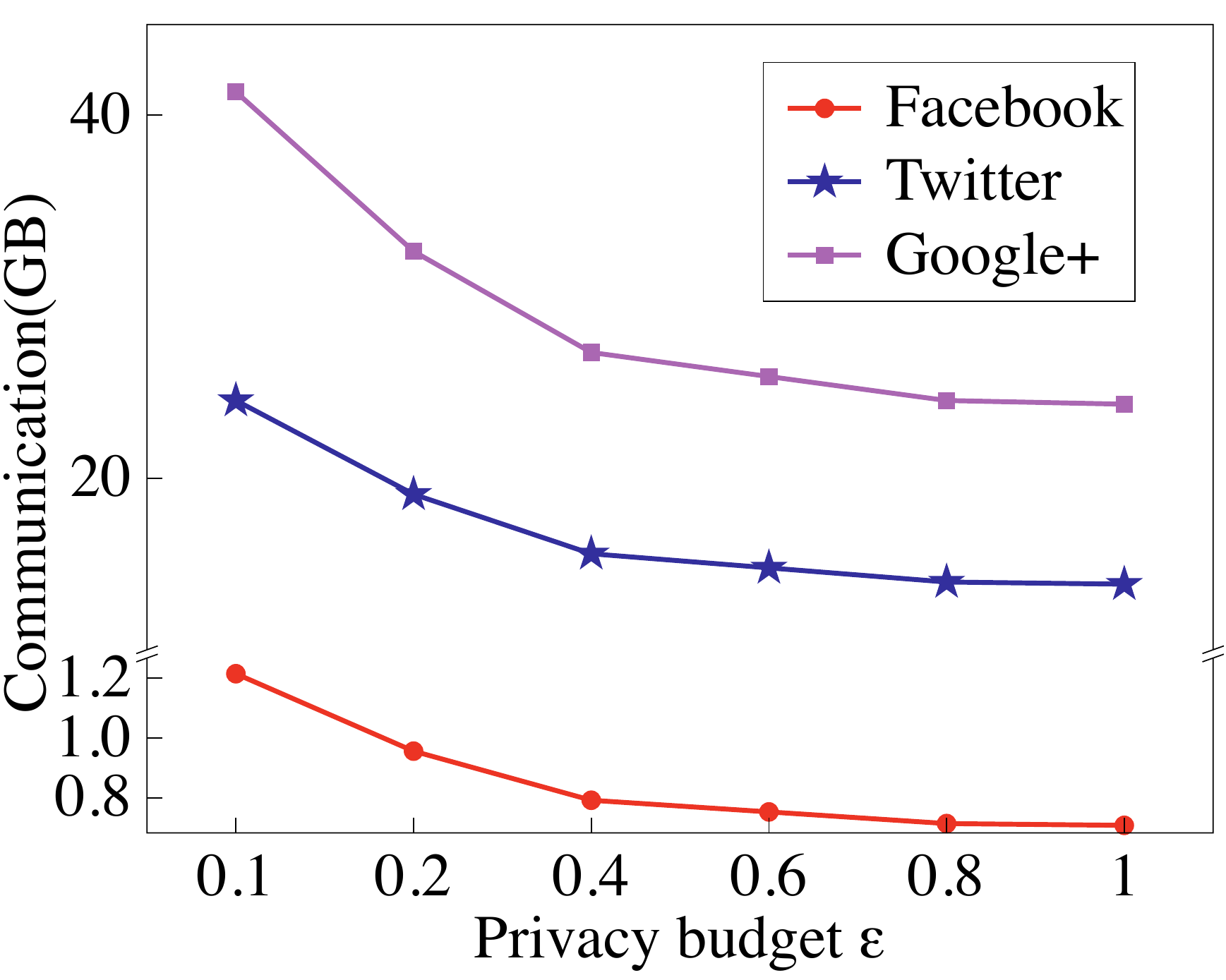}}
\end{minipage}
\caption{Running time and online communication of secure {\eig} at the cloud, for varying privacy budget $\epsilon$ ($\delta=10^{-6}$, 10 bins).}
\label{fig:perforDiffEpsi}
\vspace{-15pt}
\end{figure}

\begin{table}[!t]
\small
\centering
\caption{Performance Gain of Optimized Secure QR Algorithm}
\label{Tab:OptimQR}
\begin{tabular*}{\hsize}{@{}@{\extracolsep{\fill}}c|ccc|ccc@{}}
\toprule
&\multicolumn{3}{c|}{Time (ms)}&\multicolumn{3}{c}{Online Comm. (KB)}\\\hline

& Basic &Optim. &Gain& Basic &Optim.&Gain\\ \hline
$M=15$&341&339&0.6\% &2.3&0.2&91\%\\ \hline

$M=30$&726& 698&3.9\% &19&0.6&96\%\\ \hline

$M=45$&1172&1062&9.3\% &65&1.5&97\%\\ 
\bottomrule
\end{tabular*}
\end{table}

\noindent\textbf{Secure {\eig}.} We now report {\main}'s computation and communication performance in secure {\eig}. We first examine the performance gain of our optimized secure QR algorithm, and summarize the results in Table \ref{Tab:OptimQR}, where we set the matrix dimension output by Arnoldi method and Lanczos method to $M\times M$ ($M\in\{15,30,45\}$). Our optimized secure QR algorithm can save up to $\mathbf{97\%}$ online communication as well as $\mathbf{9.3\%}$ computation. Meanwhile, the performance gain increases as $M$ grows. We then evaluate the overall performance of secure {\eig}, and present the results in Fig. \ref{fig:perforDiffBin} and Fig. \ref{fig:perforDiffEpsi}. 
From the results, we can observe that compared to the encryption without sparse representation, {\main} can save up to $\mathbf{80\%}$ online communication as well as $\mathbf{50\%}$ computation time. Meanwhile, compared to the encryption without binning, {\main} can still save up to $\mathbf{65\%}$ online communication and $\mathbf{35\%}$ computation time.

	\begin{figure}[!t]
	\centering
	\begin{minipage}[t]{0.42\linewidth}
		\centering{\includegraphics[width=\linewidth]{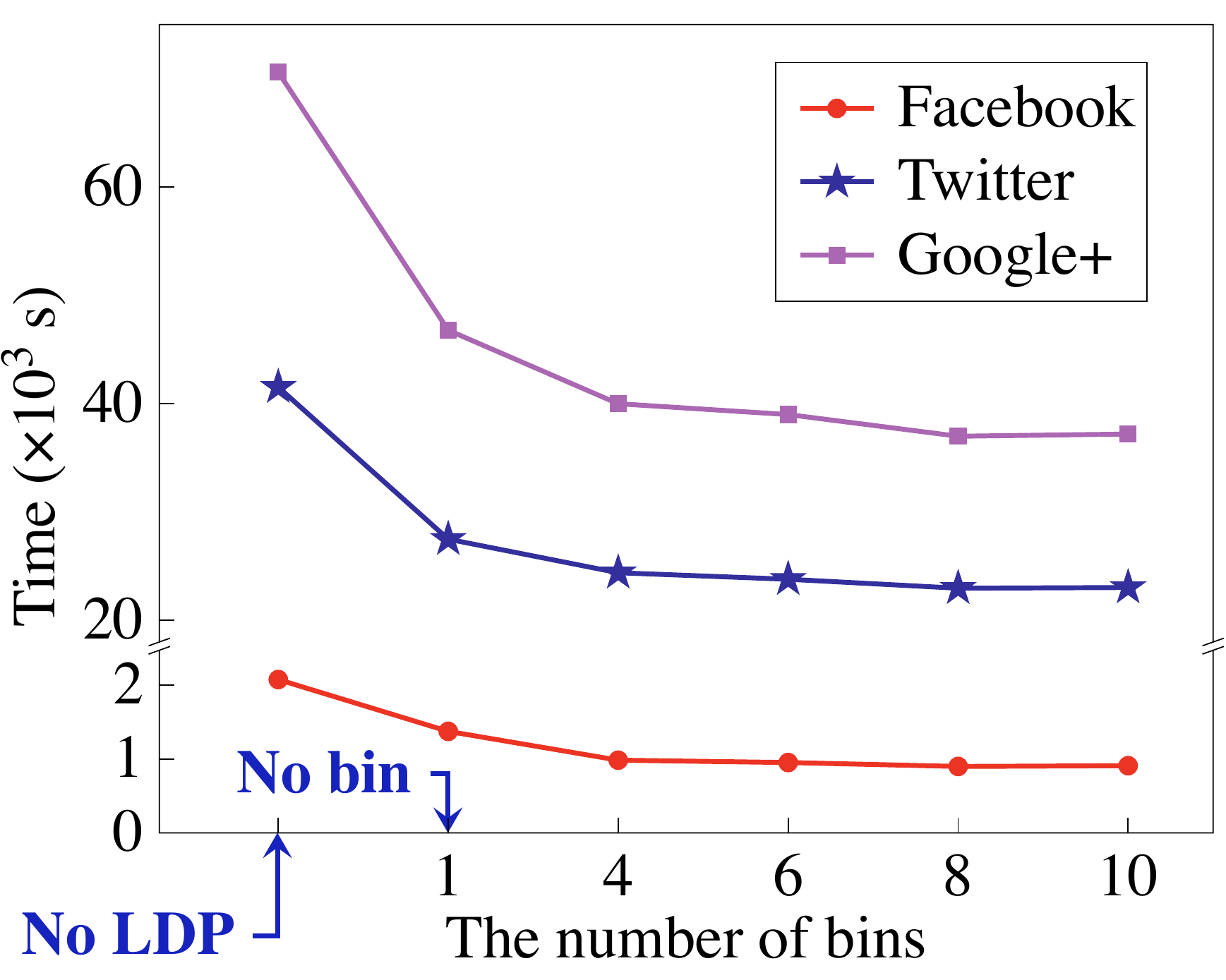}} 
	\end{minipage}
	\begin{minipage}[t]{0.42\linewidth}
		\centering{\includegraphics[width=\linewidth]{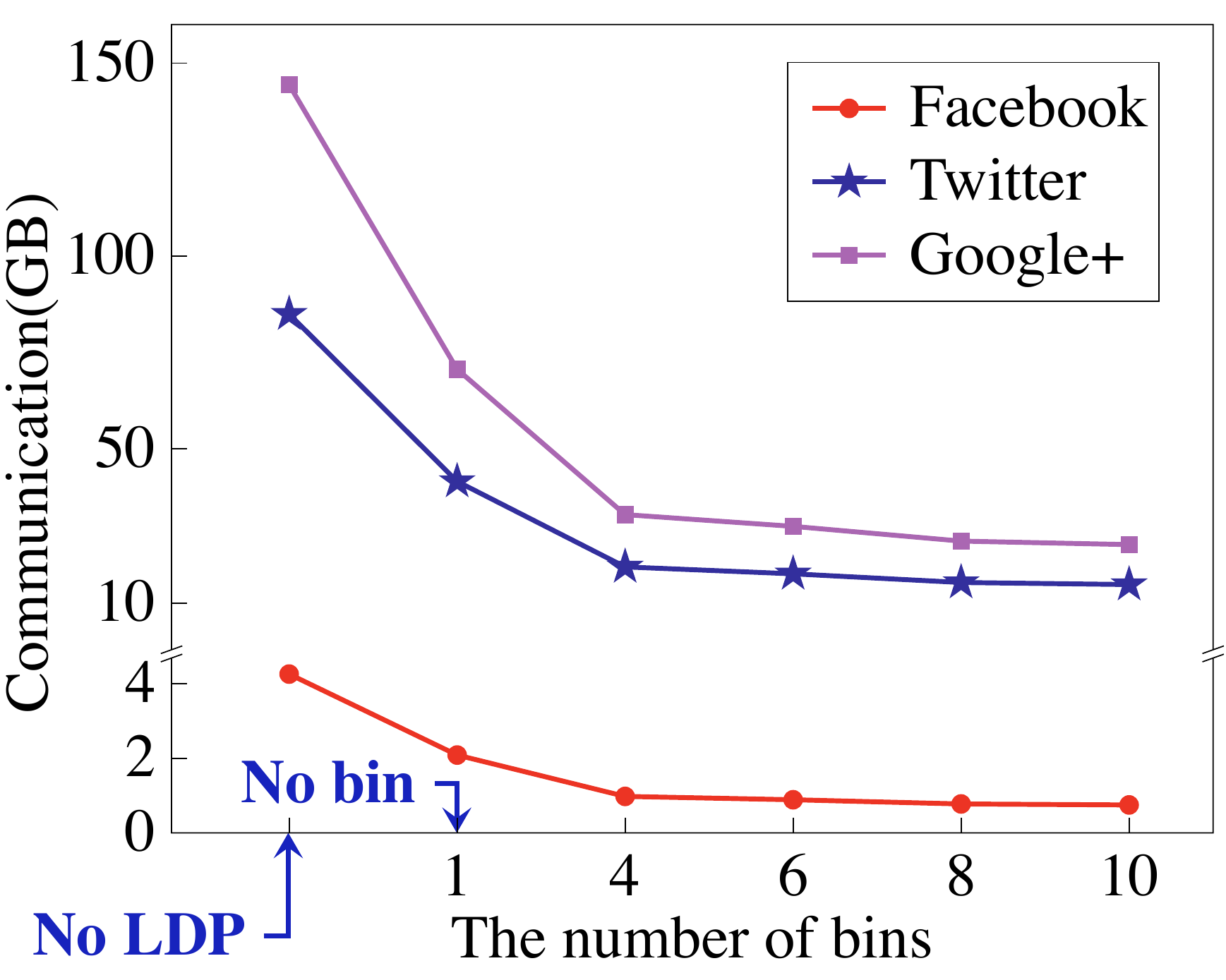}} 
	\end{minipage}
	\caption{Overall running time and online communication at the cloud, for varying number of bins ($\epsilon=1,\delta=10^{-6}$). ``No bin" refers to the use of standard LDP.}
	\label{fig:overallCost}
	\vspace{-15pt}
\end{figure}

\noindent\textbf{Overall performance.} We show in Fig. \ref{fig:overallCost} the overall running time and online communication of {\main}. It is noted that the communication results in Fig. \ref{fig:overallCost} are similar to that in Fig. \ref{fig:perforDiffBin}. This is because secure {\eig} dominates the online communication at the cloud in {\main}. Compared to the way of encryption without sparse representation, {\main} can save up to $\mathbf{80\%}$ overall online communication as well as $\mathbf{40\%}$ overall computation time. Meanwhile, compared to the way of standard LDP based encryption, {\main} can still save up to $\mathbf{65\%}$ overall online communication and $\mathbf{30\%}$ overall running time.

\vspace{-10pt}

\subsection{Comparison with the State-of-the-Art Prior Work}
%
A fair comparison between {\main} and the state-of-the-art prior work PrivateGraph \cite{sharma2018privategraph} does not exist due to its downsides analyzed in Section \ref{sec:relet1}. Specifically, PrivateGraph has limited security since it generates the binning map in plaintext domain. Besides, PrivateGraph requires frequent interactions between the cloud and the analyst who undertakes processing workload. As reported in PrivateGraph, to obtain the {\eigs} of the dataset Google+, the analyst must spend 0.2 hours as well as communicate 10 GB with the cloud. In contrast, {\main} allows the analyst to directly receive the final {\eigs}. We also note that under the same privacy budget $\epsilon=1$, the ciphertext size of PrivateGraph is up to 6.3 TB, but that of {\main} is only 7 GB.

	\vspace{-10pt}

\section{Conclusion}
\label{sec:conclusion}

We present {\main}, a new system allows privacy-preserving analytics over decentralized social graphs with \eig.
\main~leverages the emerging paradigm of cloud-empowered graph analytics paradigm and enables the cloud to collect individual user's local view in a privacy-friendly manner and perform \eig~to produce desired encrypted {\eigs} that can be delivered to an analyst on demand.
\main~delicately builds on the advancements on lightweight cryptographic techniques (ASS and FSS) and local differential privacy to securely embrace the operations required by \eig~analytics, yielding a customized security design.
Extensive experiments demonstrate that {\main} achieves accuracy comparable to the plaintext domain, with practically affordable performance superior to prior work.

\vspace{-10pt}

\section*{Acknowledgement}

This work was supported in part by the Guangdong Basic and Applied Basic Research Foundation (Grant No. 2021A1515110027), in part by the Shenzhen Science and Technology Program (Grant No. RCBS20210609103056041), and in part by the Australian Research Council (ARC) Discovery Project (Grant No. DP180103251).

\balance
\bibliographystyle{IEEEtran}
\bibliography{ref}

\end{document}